\documentclass{article}

\usepackage[light,easyscsl,sfmathbb,nowarning]{kpfonts}
\usepackage{fullpage}
\usepackage{url}
\usepackage[hidelinks]{hyperref}
\usepackage[utf8]{inputenc}
\usepackage[small]{caption}
\usepackage{graphicx}
\usepackage{amsmath}
\usepackage{amsthm}
\usepackage{amsfonts}
\usepackage{booktabs}
\usepackage{algorithm}
\usepackage{amssymb,mathtools}
\usepackage{bm}
\usepackage[pdftex,dvipsnames]{xcolor}
\usepackage{authblk}
\usepackage{xspace}
\usepackage{comment}
\usepackage{array}
\usepackage{multirow}
\usepackage{nicefrac}
\usepackage[noend]{algpseudocode}
\usepackage[numbers,sort]{natbib}

\usepackage[british]{babel}
\usepackage{hyperref}
\hypersetup{
    colorlinks=true,
    linkcolor=olive,
    citecolor=purple
}

\newtheorem{theorem}{Theorem}
\newtheorem{corollary}[theorem]{Corollary}

\newtheorem{lemma}[theorem]{Lemma}

\theoremstyle{definition}
\newtheorem{definition}[theorem]{Definition}
\newtheorem{example}[theorem]{Example}
\newtheorem{remark}[theorem]{Remark}
\newtheorem{openquestion}[theorem]{Open Question}

\newcommand{\is}{{\sc Independent Set}\xspace}
\newcommand{\maxis}{{\sc Maximum Independent Set}\xspace}
\newcommand{\issmall}{{\sc IS}\xspace}
\newcommand{\maxissmall}{{\sc MaxIS}\xspace}

\newcommand{\gpb}{{\sc Group-PB}\xspace}
\newcommand{\mgpb}{{\sc Max-Group-PB}\xspace}
\newcommand{\hpb}{{\sc Hierarchical-PB}\xspace}
\newcommand{\ugpb}{{\sc Utility-Group-PB}\xspace}
\newcommand{\arcsupply}{{\sc Arc Supply}\xspace}

\newcommand{\ddkfull}{{\sc $d$-Dimensional Knapsack}\xspace}
\newcommand{\ddk}{{\sc $d$-DK}\xspace}
\newcommand{\xdk}[1]{{\sc $#1$-DK}\xspace}

\def\calV{\mathcal{V}}
\def\calF{\mathcal{F}}
\def\calE{\mathcal{E}}
\def\calI{\mathcal{I}}

\newlength{\lengthquestion}
\newcommand{\defproblem}[3]{
  \vspace{1mm}
\noindent\fbox{
  \begin{minipage}{0.97\textwidth}
  \begin{tabular*}{\textwidth}{@{\extracolsep{\fill}}lr} #1 \\ \end{tabular*}
  \begin{tabular*}{\textwidth}{p{\lengthquestion}p{0.95\textwidth-\lengthquestion}}
    { \hfill\bf{Input:}}    & #2\\
    {       \bf{Question:}} & #3
  \end{tabular*}
  \end{minipage}
  }
  \vspace{1mm}
}

\DeclareMathOperator*{\cost}{cost}
\DeclareMathOperator*{\bundle}{bundle}
\DeclareMathOperator*{\proj}{proj}
\DeclareMathOperator*{\integer}{int}

\newcommand{\Oh}{\ensuremath{\mathcal{O}}}
\newcommand{\Ohstar}{\ensuremath{\mathcal{O}^*}}
\newcommand{\app}{\ensuremath{\mathtt{app}}}

\newcommand{\ComplexityFont}[1]{{\ensuremath{\mathsf{#1}}}}
\newcommand{\FPT}{\ComplexityFont{FPT}}
\newcommand{\NP}{\ComplexityFont{NP}}
\renewcommand{\P}{\ComplexityFont{P}}
\newcommand{\W}{\ComplexityFont{W}}
\newcommand{\XP}{\ComplexityFont{XP}}
\newcommand{\APX}{\ComplexityFont{APX}}

\title{Participatory Budgeting with Project Groups\footnote{An extended abstract of this work appears in IJCAI 2021~\cite{gpb-ijcai21}.
This is the author's version of the paper published in JCSS~\cite{gpb-jcss}; it differs only in formatting and minor editorial changes.}}

\date{}
\author[1]{Pallavi Jain\thanks{pallavi@iitj.ac.in}}
\author[2]{Krzysztof Sornat\thanks{sornat@agh.edu.pl}}
\author[3]{Nimrod Talmon\thanks{talmonn@bgu.ac.il}}
\author[3]{Meirav Zehavi\thanks{meiravze@bgu.ac.il}}
\affil[1]{Indian Institute of Technology Jodhpur, India}
\affil[2]{AGH University, Poland}
\affil[3]{Ben-Gurion University, Israel}

\begin{document}

\maketitle

\begin{abstract}
We study a generalization of the standard approval-based model of participatory budgeting (PB), in which voters are providing approval ballots over a set of predefined projects and---in addition to a global budget limit, there are several groupings of the projects, each group with its own budget limit. We study the computational complexity of identifying project bundles that maximize voter satisfaction while respecting all budget limits. We show that the problem is generally intractable and describe efficient exact algorithms for several special cases, including instances with only few groups and instances where the group structure is close to be hierarchical, as well as efficient approximation algorithms. Our results could allow, e.g., municipalities to hold richer PB processes that are thematically and geographically inclusive.
\end{abstract}

\section{Introduction}

In the standard approval-based model of participatory budgeting (PB)~\cite{cabannes2004participatory,shah2007participatory}, specifically, the model of \emph{Combinatorial PB}~\cite{aziz2021participatory}, we are given a set of $m$ projects, each with its cost, $n$ approval votes (i.e., each voter provides a subset of projects she approves), and a budget limit. The task is to aggregate the votes to select a bundle (i.e., a subset) of projects that respects the budget limit.
PB has caught quite considerable attention lately~\cite{aziz2018proportionally,talmon2019framework,aziz2021participatory,rey2023computational} as it is being used around the world to decide upon the spending of public (mostly municipal) money.
It is currently a very active topic of research in computational social choice (see, e.g., the recent survey~\cite{rey2023computational}).
Furthermore, this line of research is beginning to have practical impact with municipalities starting to use cutting-edge aggregation algorithms developed by the research community (see, e.g., the Method of Equal Shares~\cite{faliszewski2023participatory}).

Here, we consider a setting of participatory budgeting in which the projects are classified into groups that might be intersecting, and each group comes with its own budget constraint, so that the result of the PB---i.e., the aggregated bundle---shall respect not only the global budget limit, but also the limits of each of the groups.
To make things concrete and formal, below is the decision version of our problem (in the optimization version of \gpb, denoted \mgpb, the goal is to maximize the utility):

\defproblem{\gpb}
{A set $P$ of projects with their cost function $c \colon P \rightarrow \mathbb{N}$,
a set $\calV$ of voters with their approval ballots $\calE=\{P_v\subseteq P \colon v\in \calV\}$, a family of groups of projects $\calF \subseteq 2^P$ with their budget function $b \colon \calF\rightarrow \mathbb{N}$, a global budget limit $B$, and a desired utility value $u$.}
{Is there a set of projects $X \subseteq P$ such that $\sum_{v \in \calV}|P_v \cap X| \geq u$, $\sum_{p\in X}c(p)\leq B$, and for every set $F \in \calF$, $\sum_{p\in F\cap X}c(p)\leq b(F)$?}

Without loss of generality, we assume $b(F) \leq B$ for every $F \in \calF$, and that no two sets in $\calF$ are identical, i.e., $\forall_{S_1,S_2 \in \calF} \quad S_1 \neq S_2$.
Also, without loss of generality, we assume that every project in $P$ is approved by at least one voter.
(Note that voter utility equals the number of approved projects that are funded.
This is the most popular definition of utility in PB.
However, other definitions exist, some of which we discuss later.)

\begin{example}
Let $P=\{p_1,p_2,p_3,p_4\}$. Let the cost of projects be as follows:
$c(p_1)=2,c(p_2)=1, c(p_3)=3, c(p_4)=1$.
Suppose that we have only $2$ voters, say $v,v'$ and $P_v=\{p_1,p_2,p_3\}, P_{v'}=\{p_3,p_4\}$.
Let $\calF = \{F_1=\{p_1,p_3\}, F_2=\{p_2,p_4\}\}$. Let $b(F_1)=3, b(F_2)=2$. Note that we are allowed to take only one project from $F_1$. Let $B=5$ and $u=3$. Let $X=\{p_3,p_4\}$.
Note that $|P_v \cap X|=1$ and $|P_{v'}\cap X|=2$.
Thus, $\sum_{v \in \calV}|P_v \cap X|= 3 =u$.
Further, $c(p_3)+c(p_4)=4\leq B$, and the cost of projects from sets $F_1$ and $F_2$ are $3$ and  $1$, respectively, that is,
$\sum_{p\in F_1}c(p) =3= b(F_1)$ and $\sum_{p\in F_2}c(p)=1 \leq b(F_2)$.
\end{example}

Our work is motivated mainly by the following use-cases:
\begin{enumerate}

\item
\textbf{Geographical budgeting:} Consider a city (say, Paris), consisting of several districts.
To not spend all public funds on projects from, say, only one district, \gpb is useful:
group projects according by districts and select appropriate budget limits (making sure that, e.g., none of Paris's $20$ districts would use more than $10\%$ of the total budget).
(A more fine-grained solution, incorporating neighborhoods, streets, etc., is also possible.
Currently, such geographic inclusiveness is usually achieved ad hoc by holding separate per-district elections).

\item
\textbf{Thematic budgeting:}
Projects usually can be naturally grouped into types, e.g., education projects,  recreational projects, and so on.
\gpb is useful here: group projects accordingly, making sure that not all the budget is being spent on, say, projects of only one type.

\item
\textbf{Non-budgeting use-cases:}
\gpb is useful in contexts other than PB: e.g., to decide which processes to run on a time-limited computing server, where available processes can be naturally grouped into types and it is not desired to use all computing power for, say, processes of only one type.
In this setting, the voters may be different users with different requirements for processes to be finished in a machine run.

\end{enumerate}

\begin{remark}
Throughout the paper we concentrate on different groups and on the topology and structure of such groups. One particular interesting topology is where the groups are hierarchical. This comes naturally in many settings, for example (corresponding to the motivating use-cases described above):
\begin{enumerate}

\item
Geographical arrangements of projects often manifest hierarchy.
It is evident, e.g., in the organization of urban areas into cities, districts, neighborhoods, and streets.

\item
Thematic groupings of projects, such as to recreational, health-related, and education-related projects, often exhibit a hierarchical structure, where overarching categories encompass subcategories, leading to a more granular and organized representation of diverse initiatives.

\item
Organizations commonly adopt a hierarchical structure---being composed of groups, teams, smaller units, to facilitate efficient management and collaboration.

\end{enumerate}
\end{remark}

\begin{table}[t]
    \centering
    \begin{tabular}{l l l}
\toprule
\textbf{Result} & \textbf{Parameters} & \textbf{Reference} \\
\midrule
para$\NP$-hard & $b_{\max}+s+\app+\ell$ & Theorem~\ref{thm:paraNP-l+s+f} \\
para$\NP$-hard & $b_{\max}+s+n+\ell$ & Theorem~\ref{thm:paraNP-l+s+f} \\
$\W[1]$-hard & $b_{\max}+s+\app+u$ & Theorem~\ref{thm:w-hard-app2} \\
$\W[1]$-hard & $b_{\max}+s+n+u$ & Theorem~\ref{thm:w-hard-app2} \\
para$\NP$-hard & $D_g+\ell$ & Theorem~\ref{thm:partition-pb-xp-dg} \\
$\FPT$ & $D_g+s$ & Theorem~\ref{thm:partition-pb-fpt-dg+s} \\
$\FPT$ & $D_p$ & Theorem~\ref{thm:partition-pb-fpt-d} \\
$\XP$ & $g$ &
Theorem~\ref{thm:xp-g} \\
$\FPT$ & $g + u$ &
Theorem~\ref{thm:fpt-u+g} \\
$\FPT$ & $g + n$ &
Theorem~\ref{thm:fpt-g+n} \\
$\FPT$ & $g + B$ &
Theorem~\ref{thm:fpt-g+b} \\
\bottomrule
%
%
    \end{tabular}
    \caption{Parameterized complexity of \gpb wrt.\ the layerwidth $\ell$ (see Definition~\ref{def:layerwidth}), $b_{\max}=\max\{b(F): F \in \cal{F}\}$, $s=\max\{|F|: F \in \calF\}$, the maximum number $\app$ of projects a voter approves, the number $n$ of voters, the required utility $u$, the number $D_g$ ($D_p$) of groups (resp., projects) to delete to get a hierarchical structure, the number of groups $g=|\calF|$, and the total budget $B$.
    }
    \label{tab:pc-results}
\end{table}

\subsection{Our Contribution}

We introduce and study \gpb,
first demonstrating its computational intractability even for some very restricted cases (Theorems~\ref{thm:paraNP-l+s+f} and~\ref{thm:partition-pb-xp-dg}).
Interestingly, \gpb can be solved in polynomial-time if a project belongs to at most one group, but becomes $\NP$-hard as soon as a project can belong to two groups (Theorem~\ref{thm:partition-pb-xp-dg}).
Positively, we show that \gpb can be solved in polynomial-time for a constant number of groups (Theorem~\ref{thm:xp-g}) and
for instances with hierarchical group structure (i.e., any pair of groups must be either non-intersecting or in containment relation; Lemma~\ref{lemma:hier-pb-in-p}).
Note that in some cases, e.g., when grouping projects by geographical regions/districts/neighborhoods, the group structure is indeed hierarchical.

We extend our study to approximation and parameterized algorithms.
We design efficient approximation algorithms for some cases (Theorem~\ref{thm:ptas-for-g-loglog-i}), and consider the following parameters:
  the number of projects ($m$), the number of voters ($n$), the maximum size of an approval set ($\app$), the budget ($B$), the maximum size of a group in the family $\calF$ ($s$), the utility ($u$), the layerwidth ($\ell$) (see Definition~\ref{def:layerwidth} for layerwidth), the size of the family $\calF$ ($g$), and $b_{\max}=\max_{F\in \calF}b(F)$; and obtain both tractability and intractability results.
The motivation for considering these is the following:
  the number of voters, $n$, can be small in cases when we do PB by the council or in a small community;
  $m$ is sometimes quite small (e.g., the PB instances of Stanford Participatory Budgeting Platform\footnote{\url{https://pbstanford.org/}} usually consist of only 10--20 projects);
  $\ell$ and $b_{\max}$ can be set by designer---they are generally not small, yet we add them for completeness;
  $B$ and $u$ are also generally not small, but added for completeness;
  $\app,s$ and $g$ can be set by the designer, and they are usually rather small, e.g.~$5$ to $10$.

  Since the problem can be solved in polynomial-time on hierarchical instances, we also consider two distance parameters to a hierarchical instance, $D_p$ and $D_g$, the minimum number of projects (respectively, groups) whose deletion leads to a hierarchical instance.
  Finding efficient algorithms for such distance parameters implies that not only hierarchical instances can be solved efficiently, but also instances that are close to being such (Theorems~\ref{thm:partition-pb-fpt-dg+s} and~\ref{thm:partition-pb-fpt-d}).
  This is particularly useful in the presence of a few outliers (due to this, distance parameters are studied frequently in parameterized complexity, including in voting theory~\cite{bredereck2014parameterized,GuptaJRSZ20,wellstructured}).

  In particular, the main focus of our paper is on adding a group structure on top of a standard PB instance.
  From this point of view, PB election designers can choose how complex they want the group structure to be. Thus, studying the complexity of \gpb wrt.\ our parameters---in particular, the parameter $g$ and the distance parameters---sheds light on the effect of adding groups on the complexity of the problem (which is polynomial-time solvable when there are no groups as it can be solved using dynamic programming via equivalence to {\sc Unary Knapsack}~\cite{talmon2019framework}).
  Following our parameterized tractability results, a PB election designer can practically use our group structures, albeit perhaps not with an arbitrary number of groups of unlimited structural complexity.

  Table~\ref{tab:pc-results} lists most of our complexity results.
  Parameterized complexity wrt.\ $g$ is open (Open Question~\ref{thm:fpt-g}), however we have an approximation scheme that is $\FPT$ wrt.\ $g$ (Theorem~\ref{thm:fpt-as-simpler}) as well as $\W[1]$-hardness wrt.\ $g$ for a slightly more general problem (Theorem~\ref{thm:w-h-ugpb}).
  Tables~\ref{tab:apx-results} and~\ref{tab:inapx-results} summarize our (in)approx\-imability results.

\begin{remark}
While some city planners may not care for complexity results, we are personally aware of some that are hesitant to use algorithmic methods that may not be efficient and thus may require extensive computational resources.\footnote{In particular, one of the authors, while trying to convince a deputy mayor of a medium-sized city to implement a participatory budgeting process, faced significant criticism regarding the worry of the need of using extensive computational resources.}
Thus, in addition to being of theoretical interest, our complexity analysis results have practical implications regarding the feasibility of adding group-wise budget upper bounds to PB. (At least, as much as theoretical results imply practical feasibility.)
\end{remark}

\begin{table}[t]
    \centering
    \smallskip
    \begin{tabular}{lll}
        \toprule
        {\bf Upper-bound on $g$} & {\bf Approximation ratio} & {\bf Reference} \\
        \midrule
        $\Oh(1)$ & P & Corollary~\ref{cor:g-constant-poly}\\
        $\log_4\log(m^{\Oh(1)})$ & PTAS & Corollary~\ref{cor:ptas-for-g-loglogm}\\
        any $g$ & $g+2$ & Corollary~\ref{cor:g-plus-2-apx}\\
        \bottomrule
    \end{tabular}
    \caption{Achieved approximation ratios (in polynomial time) depending on the number $g$ of groups.
    The smaller $g$ compared to $m$, the better approximation for \mgpb we can achieve.} \label{tab:apx-results}
\end{table}
\begin{table}[t]
    \centering
    \smallskip
    \begin{tabular}{lll}
        \toprule
        {\bf Lower-bound on $g$} & {\bf Inapproximability} & {\bf Reference}  \\
        \midrule
        $\frac{3}{2}m$ & no PTAS & Theorem~\ref{thm:no-ptas-g-linear-on-m}\\
        $m^2$ & no $g^{\frac{1}{2}-\epsilon}$-approximation alg. & Theorem~\ref{thm:sqrt-g-apx}\\
        $m^2$ & no $m^{1-\epsilon}$-approximation alg. & Theorem~\ref{thm:sqrt-g-apx}\\
        \bottomrule
    \end{tabular}
    \caption{Achieved polynomial-time inapproximability results depending on the number $g$ of groups.
    The larger $g$ compared to $m$, the higher is the approximation ratio excluded.} \label{tab:inapx-results}
\end{table}

\paragraph{Initial Observations.}

For completeness, we mention that \gpb is trivially $\FPT$ wrt.\ $m$, by a brute-force algorithm in $\Ohstar(2^m)$ time\footnote{\Ohstar{} hides factors that are polynomial in the input size.} (and, as the Exponential Time Hypothesis implies a lower bound of $2^{o(|V|)}$ for \is, we conclude a lower bound of $2^{o(m)}$ following the reduction in the proof of Theorem~\ref{thm:paraNP-l+s+f}).
Furthermore, \gpb is $\FPT$ wrt.\ $\app+n$ as every project is approved by at least one voter, implying $m \leq \app\cdot n$.

\paragraph{Road Map.}
In Subsection~\ref{subsec:related-works} we discuss research the most related to ours.
In Section~\ref{sec:layerdecompositions} we consider a structural property of family of sets, useful for obtaining a polynomial-time algorithm for hierarchical families and may also be of independent interest.
Then, in Section~\ref{sec:intractability}, we present intractability results of \gpb. Sections~\ref{sec:tractability1} and~\ref{sec:tractability2} are devoted to parameterized analysis of \gpb.
The approximability results are in Section~\ref{sec:apx}.
We conclude in Section~\ref{sec:outlook}, in which we discuss future research directions.

\subsection{Related Work}\label{subsec:related-works}
The literature on PB is quite rich~\cite{aziz2018proportionally,rey2023computational}; formally, we generalize the framework of Talmon and Faliszewski~\cite{talmon2019framework} by adding group structures to approval-based PB.
Jain \textit{et al.}~\cite{pbsub} and Patel \textit{et al.}~\cite{PatelKL21} also consider---albeit significantly simpler---group structures (with {\it layerwidth} $1$; see Definition~\ref{def:layerwidth}).

Fairness constraints are studied, e.g., in the contexts of influence maximization~\cite{tsang2019group}, clustering~\cite{chierichetti2017fair}, and allocation problems~\cite{benabbou2018diversity}.
Our focus is on fairness in PB (e.g., not spending all funds on one district). Recently, Hershkowitz \textit{et al.}~\cite{HershkowitzKPP21} introduced a district-fairness notion by allowing projects have different utility for different districts. There are papers on fairness and group structures for the special model of multiwinner elections~\cite{izsak2018committee,CelisHV18,yang2018multiwinner,wellstructured,ianovski2022electing}.

Technically, \gpb is a special case of the \ddkfull problem (\ddk; also called  Multidimensional Knapsack)~\cite[Section 9]{knapsackbook}:
  given a set of items, each having a  $d$-dimensional size-vector and its utility, a $d$-dimensional knapsack capacity vector $\beta$ with an entry for each dimension, and required integer utility---with all input numbers being non-negative integers---the goal is to choose a subset of the items with at least the required total utility and such that the sum of the chosen items' sizes is bounded
by the knapsack capacity, in each dimension.
\ddk generalizes \gpb:
  items in \ddk correspond to projects; fix an order on $\calF$, i.e., $(F_1,F_2,\dots,F_g)$, resulting in $d = g+1$ many dimensions, a $(g+1)$-dimensional size vector $\gamma$ for an item $p \in P$, defined by
$\gamma_p(i) = c(p)$ if $p \in F_i$ and $\gamma_p(i) = 0$ otherwise, $\gamma_p(g+1) = c(p)$ for every $p\in P$, corresponding to a global budget, utility of an item $p \in P$ equals its approval score, required utility in \ddk equals $u$, and the $(g+1)$-dimensional bin $\beta$ is defined via $\beta(i) = b(F_i)$ for $i \in \{1,2,\dots,g\}$, with $\beta(g+1) = B$.
So, \gpb is an instance of \xdk{(g+1)} where each item $p \in P$ has only two possible sizes over dimensions, i.e., $0$ and $c(p)$.
Crucially, as our model is a special case we can use our special instance structure; hence, we treat results for \ddk as a good benchmark, in particular, the (in)approximability results for \ddk~\cite{knapsackbook}.

Recently, many papers introducing fairness components to well known computational problems have appeared.
One of them models fairness among voters by introducing a more general utility function for the voters.
In essence, instead of a linear function one may apply any function.
For example, a harmonic utility function $f(q) = \sum_{q=1}^i 1/q$ (or any concave) models the law of diminishing marginal utility known in economics.
Having unit-cost projects and harmonic utility function we get the {\sc Proportional Approval Voting} problem~\cite{Thiele95,Kilgour2010,ByrkaSS18,BarmanFF21,DudyczMMS20}.
When additionally allowing projects to have different costs, we fall into {\sc Fair Knapsack}~\cite{fluschnik2019fair}.
Note that these problems model fairness by specifying what is the total utility of an outcome.
In essence, a voting rule with properly chosen utility functions disallows
a large group of voters to decide a final outcome being their most preferred projects.
Hence preferences of a small group impact a solution as well.
A more general framework was proposed by Jain \textit{et al.}~\cite{pbsub} in which
projects are partitioned into groups and
a utility function can be a general non-decreasing function.
For each voter a utility function is applied for each group of projects separately.
This models interaction between projects.
The most interesting utility functions are concave (and convex) ones
that model substitution (and complementarity respectively) effects for groups of projects.
For such functions Jain \textit{et al.}~\cite{pbsub}
achieved more positive algorithmic results (polynomial-time or fixed-parameter tractable (FPT) algorithms)
in contrast to general utility functions
that are not possible to approximate up to a factor better than $n^{o(1)}$ (assuming Gap-ETH).
For concave utility functions also utility-lost minimization variants where considered~\cite{SwamyS08,ByrkaSS18}, motivated by some fairness notion in facility location context.

In contrast to the works described in the previous paragraph,
in this paper we focus on the most natural and the simplest linear utility function, which has practical advantages,
but to achieve fairness we group projects and add budget-limits for each group, i.e.,
every feasible outcome does not exceed a budget-limit among the given groups of projects.
A project may be contained in a few or none of the groups.
Therefore, a large group of voters cannot force choosing only one type of projects (educational, environmental etc.)
that is supported by the large group.

There is a recent paper~\cite{PatelKL21} which introduce categories for the projects, and consider additional fairness notions among the categories.
First of all, they allow to have both a lower-bound and an upper-bound on the total cost of the projects taken from each category.
As a separate model they propose to have a lower-bound and an upper-bound on the cardinality of the set of projects taken from each category.
The third model makes use of lower-bounds and upper-bounds on the total utility achieved by the projects taken from each category.
Further, instead of approval ballots they consider general values for utilities (they call it {\it values}) and general costs of projects, i.e., they can be any non-negative real numbers.
Therefore it may seem that Patel \textit{et al.}~\cite{PatelKL21} considered a much more general model than we do, but there is a very crucial difference:
their categories of projects make a partition of projects, i.e., each project belongs to exactly one category (similarly to a project interaction model introduced by Jain \textit{et al.}~\cite{pbsub}).
Our model allows to have a family of groups of projects $\calF$ that is any subset of $2^P$.
It means we could have even $2^{|P|}$ many groups in an instance
in contrast to the model of Patel \textit{et al.}~\cite{PatelKL21} that can have at most $|P|$ groups.
It shows that the model of Patel \textit{et al.}~\cite{PatelKL21} and our model have a non-empty intersection and they generalize the intersection in different directions.
In particular, the model proposed by Patel \textit{et al.}~\cite{PatelKL21} with only cost upper-bounds on categories is a special case of our model that has {\it layerwidth} (see Definition~\ref{def:layerwidth}) equal to $1$ and
which is a special case of {\it hierarchical} instances that we show how to solve in polynomial time.
In contrast to our results, the results presented by Patel \textit{et al.}~\cite{PatelKL21} for cost upper-bounds provide the required total utility (an exact solution) but violate group cost bounds and a global budget by a fixed factor $\epsilon$.
They rely on a standard bucketing technique extensively, which was used to construct a fully polynomial-time approximation scheme (FPTAS) for the classical knapsack problem~\cite{IbarraK75,Lawler79}.
This limits the number of different costs and also causes possible bound violations.
Our results do not violate any of the constraints.
This is possible because the utilities of the projects come from approval ballots hence they are integral and (we may think) they are decoded in unary.
In contrast to the work of Patel \textit{et al.}~\cite{PatelKL21}, which considers
approximation algorithms, we consider also parameterized complexity (hence exact feasible solutions).

\subsection{Preliminaries}

We use the notation $[n] = \{1,2,\dots,n\}$ for $n \in \mathbb{N}$.

In parameterized complexity, problem instances are associated with a {\em parameter}, so that a \emph{parameterized problem} $\rm \Pi$ is a subset of $\Sigma^{*} \times \mathbb{N}$, where $\Sigma$ is a finite alphabet. A corresponding instance is a tuple $(x,k)$, where $x$ is a classical problem instance and $k$ is the parameter. A central notion in parameterized complexity is \emph{fixed-parameter tractability} (\FPT, in short), which means, for a given instance $(x,k)$, decidability in $f(k) \cdot {\sf poly}(|x|)$ time, where $f(\cdot)$ is an arbitrary computable function and ${\sf poly}(\cdot)$ is a polynomial function. The Exponential Time Hypothesis (ETH) is a conjecture stating that $3$-{\sc SAT} cannot be solved in time that is subexponential in the number of variables. For more about parameterized complexity see, e.g., the book of Cygan \textit{et al.}~\cite{CyganFKLMPPS15}.

\section{Layer Decompositions}\label{sec:layerdecompositions}

The following is a useful structural property.
The results of this section would be useful throughout the paper.

\begin{definition}[Layer Decomposition]
  A {\em layer decomposition} of a family of sets $\calF$ is a partition of the sets in $\calF$ such that every two sets in a part are disjoint.
  Each part is a {\em layer}.
\end{definition}

\begin{definition}[Layerwidth $\ell$]\label{def:layerwidth}
The {\em width} of a layer decomposition is the number of layers in it.
The {\em layerwidth} of a family of sets $\calF$, denoted by $\ell(\calF)$ (or simply $\ell$ if $\calF$ is clear from the context), is the minimum width among all possible layer decompositions of $\calF$.
\end{definition}

We consider the problem of finding a layer decomposition, and refer to the problem of finding one of minimum width as {\sc Min Layer Decomposition}, with its decision version called {\sc Layer Decomposition}, which, given a family $\calF$ of sets and an integer $\ell$, asks for the existence of a layer decomposition with width~$\ell$.

Unfortunately, {\sc Layer Decomposition} is intractable via a reduction from {\sc Edge Coloring}.

\begin{theorem}\label{thm:LD-hard}
  {\sc Layer Decomposition} is $\NP$-hard even when $\ell=3$ and $s=2$, where $s$ is the maximum size of a set in the given family $\calF$.
\end{theorem}

\begin{proof}
We provide a polynomial time reduction from the {\sc Edge Coloring} problem, in which we are given a graph $G$ and an integer $k$;
we shall decide the existence of a mapping $\Psi\colon E(G)\rightarrow [k]$ such that if two edges $e,e' \in E(G)$ share an endpoint, then $\Psi(e)\neq \Psi(e')$.
This problem, even when $k=3$, is known to be $\NP$-hard~\cite{holyer1981np}. We call such a function a {\em proper coloring function}.
We create a family of sets, $\calF$, as follows.
For every edge $uv \in E(G)$, we have a set $\{u,v\}$ in the family $\calF$.
We set $\ell =k$.
Next, we prove the equivalence between the instances $(G,k)$ of the {\sc Edge Coloring} problem and $(\calF,\ell)$ of the {\sc Layer Decomposition} problem.

In the forward direction, let $\Psi\colon E(G)\rightarrow [k]$ be a proper coloring function.
We create a partition, $L_1,\ldots,L_\ell$, of the sets in $\cal F$ as follows:
$L_i = \{\{u,v\}\colon \Psi(uv)\in i\}$.
Due to the definition of a proper coloring function, every two sets in a part $L_i$, $i\in [\ell]$, are disjoint.

In the backward direction, let $L_1,\ldots,L_\ell$ be a layer decomposition of $\calF$.
We define proper coloring function $\Psi$ as follows: if the set $\{u,v\}\in L_i$, $i\in[k]$, then $\Psi(uv) = i$.
We claim that $\Psi$ is a proper coloring function.
Suppose not, then there are two edges $e,e'$ that share endpoints and $\Psi(e)=\Psi(e')=i$.
Without loss of generality, let $e=uv$ and $e'=uw$.
Since $\Psi(e)=\Psi(e')=i$, $\{u,v\},\{u,w\} \in L_i$, contradicting that $L_i$ is a part in a layer decomposition of~$\calF$.
\end{proof}

A reduction to {\sc $2$-Graph Coloring} gives a polynomial-time algorithm for layerwidth $2$.

\begin{theorem}
There exists a polynomial-time algorithm that finds a layer decomposition of layerwidth two, if it exists.
\end{theorem}

\begin{proof}
 We begin with giving a polynomial time reduction from {\sc Layer Decomposition} to the {\sc Vertex Coloring} problem, in which given a graph $G$ and an integer $k$, we need to  decide the existence of a mapping $\Psi\colon V(G)\rightarrow [k]$ such that if $uv\in E(G)$, then $\Psi(u)\neq \Psi(v)$.
 We call this mapping a {\em proper vertex coloring function}.
 Given an instance $({\cal F},\ell)$ of {\sc Layer Decomposition}, we create an instance $(G,k)$ of {\sc Vertex Coloring} as follows.
 For every set $F \in {\cal F}$, we add a vertex $u_F$ in the graph $G$.
 Next, we define the edge set of $G$. Let $F$ and $F'$ be two sets in ${\cal F}$ such that $F \cap F' \neq \emptyset$, then we add an edge $u_Fu_{F'}$ to $G$.
 We set $k=\ell$.
 Next, we prove the equivalence between the instance $({\cal F},\ell)$ of {\sc Layer Decomposition} and $(G,k)$ of {\sc Vertex Coloring}.
 
 In the forward direction, let $L_1,\ldots,L_\ell$ be a layer decomposition of $\calF$.
 We construct a proper vertex coloring function $\Psi$ as follows.
 If a set $F$ belongs to the set $L_i$, where $i\in [\ell]$, then $\Psi(u_F)=i$.
 We claim that if $u_Fu_{F'} \in E(G)$, then $\Psi(u_F)\neq \Psi(u_{F'})$.
 Consider an edge $u_Fu_{F'} \in E(G)$.
 Note that by the construction of the graph $G$, $F\cap F'\neq \emptyset$.
 Therefore, sets $F$ and $F'$ belong to different sets in $\{L_1,\ldots,L_\ell\}$. Hence, $\Psi(u_F)\neq \Psi(u_{F'})$.
 
 In the backward direction, let $\Psi$ be a proper vertex coloring function.
 We create a layer decomposition $L_1,\ldots,L_\ell$ as follows.
 If for a vertex $u_F$, $\Psi(u_F)=i$, then we add $F$ to layer $L_i$.
 Next, we prove that every two sets in every layer $L_i$, $i\in [\ell]$, are disjoint.
 Towards the contradiction, suppose that there exist two sets $F,F'$ in a layer $L_i$, $i\in [\ell]$ such that $F\cap F' \neq \emptyset$.
 Since $F,F'$ are in layer $L_i$, by the construction of the layer decomposition, we have that $\Psi(u_F)=\Psi(u_{F'})=i$.
 Since $F\cap F' \neq \emptyset$, by the construction of $G$, $u_Fu_{F'}\in E(G)$.
 This contradicts the fact that $\Psi$ is a proper vertex coloring function.
 
 Since we can find a proper coloring function $\Psi\colon V(G)\rightarrow [2]$ in polynomial time, if it exists, as it is equivalent to checking if the given graph is a bipartite graph (see, e.g.,~\cite{cormen}), the proof is complete.
\end{proof}

We discuss hierarchical families of sets (also known as {\it laminar families}).

\begin{definition}[Hierarchical Family]
A family of sets $\cal{F}$ is called \emph{hierarchical}, if every two sets $F_1$ and $F_2$ in the family~$\cal{F}$ are either disjoint or $F_1\subset F_2$ or $F_2\subset F_1$.
\end{definition}

\begin{theorem}\label{thm:poly-ld}
  There exists a polynomial-time algorithm that solves a given instance $({\calF},\ell)$ of {\sc Layer Decomposition} when $\calF$ is a hierarchical family.
\end{theorem}

\begin{proof}
Let $\calF'=\calF\cup \{S\}$, where $S$ is the set of all the elements which are in the sets in the family $\cal{F}$, and $\ell'=\ell+1$.
Clearly, $(\cal{F},\ell)$ is a yes-instance of {\sc Layer Decomposition} if and only if $(\cal{F}',\ell')$ is a yes-instance of {\sc Layer Decomposition}.

We first create a directed graph, $G$, called a {\em subset relationship graph}, as follows.
For every set $F \in {\calF}'$, we add a vertex $v_F$ in $G$.
Now, if $F\subset F'$, where $F,F' \in {\calF}$, then we add an arc $(v_{F'},v_{F})$ to the graph $G$.
Note that $G$ is a transitive acyclic graph.
Let $\sigma$ be a topological ordering (a topological ordering is an ordering, $\sigma$, of the vertices of a digraph such that if $(u,v)$ is an arc in the digraph, then $\sigma(u)<\sigma(v)$;
and it can be found for an acyclic graph in polynomial time~\cite{kleinberg2006algorithm}) of $G$.
Next, we create an out-tree rooted at $\sigma(1)$, say $T$, using DFS in which we first traverse the neighbor of a vertex which appears first in the ordering $\sigma$.
That is, if the currently visited vertex is $v$, and $z_1,z_2$ are out-neighbors of $v$, then we first visit $z_1$ if $\sigma(z_1)<\sigma(z_2)$, and we first visit $z_2$ otherwise. We create a partition, say $\cal{L}$, of $\calF$ as follows:
Sets corresponding to the vertices at a level of out-tree $T$ forms a layer in $\cal{L}$.
We first prove that $\cal{L}$ is a layer decomposition.
Suppose not, then there exist two sets $F_1,F_2$ in a layer in $\cal{L}$ which are not disjoint.
Without loss of generality, let $F_1 \subset F_2$.
Then, due the construction of $G$, $(v_{F_2},v_{F_1})$ is an arc in $G$.
This implies that $\sigma(F_2)<\sigma(F_1)$.
Since $\calF$ is a hierarchical family, if $F_2 \subset F$, then $F_1\subset F$.
Thus, an in-neighbor of $v_{F_1}$ is also an in-neighbor of $v_{F_2}$.
Therefore, by our construction of out-tree, we first visit $v_{F_2}$ than $v_{F_1}$.
Since $(v_{F_2},v_{F_1})$ is an arc in $G$, it contradicts that $v_{F_2}$ and $v_{F_1}$ are at the same level in the out-tree $T$ which is constructed using DFS.

Next, we argue that the width of $\cal{L}$ is $\ell'$ if and only if $(\calF', \ell')$ is a yes-instance of {\sc Layer Decomposition}.
The forward direction follows trivially. Next, we show that if the width of $\cal{L}$ is more than $\ell'$, then $(\calF', \ell')$ is a no-instance of {\sc Layer Decomposition}.
Suppose that the width of $\cal{L}$ is more than $\ell'$, then due to the construction of the layer decomposition, the number of levels in the out-tree $T$ is more than $\ell'$.
This implies that there exists a directed path from root whose length is at least $\ell'+1$.
Due to the transitivity of graph $G$, directed graph induced on the vertices in this path forms a tournament.
Thus, all the sets corresponding to vertices in this path belong to distinct layers. Thus, the layerwidth of $\calF'$ is at least $\ell'+1$.
This completes the proof.
\end{proof}

\begin{remark}\label{rem:hierarchical-ld}
 The general idea of the algorithm described in the proof of Theorem~\ref{thm:poly-ld} is to build a graph with one vertex for each group and edges corresponding to group intersections, followed by traversing the graph in topological order and constructing the corresponding hierarchical tree.
 Note that, conveniently, the algorithm can be modified to construct an ordered layer decomposition such that every set in the $i$-th layer is a subset of a set in the $(i-1)$-th layer.
\end{remark}

\section{Intractability of General Instances}\label{sec:intractability}

Next we prove intractability, showing that \gpb is $\NP$-hard even when some of the input parameters are constant.
Note, importantly, that we can solve the standard PB problem---without project groups---in polynomial time (as it can be solved using dynamic programming via equivalence to {\sc Unary Knapsack}~\cite{talmon2019framework}).

The following result is obtained via reductions from the \is (\issmall) problem on $3$-regular $3$-edge colorable graphs.

\begin{theorem}\label{thm:paraNP-l+s+f}
\gpb\ is $\NP$-complete even when $b_{\max}=1$, $s=2$, $\app=1$, and $\ell=3$; and even when $b_{\max}=1$, $s=2$, $n=1$, and $\ell=3$.
\end{theorem}

\begin{proof}

We first show the proof for $b_{\max}=1$, $s=2$, $\app=1$, and $\ell=3$; afterwards, we show how to modify it to get the other claim.

We describe a polynomial-time reduction from the \is (\issmall) problem on $3$-regular $3$-edge colorable graphs,
in which given a graph $G=(V,E)$,
where each vertex has degree $3$ and the edges of graph can be properly colored using $3$ colors (no two edges that share an end-point are colored using the same color),
and an integer $k$, we need to decide the existence of a $k$-sized independent set (there is no edge between any pair of vertices in the set).
\issmall is known to be $\NP$-complete on $3$-regular $3$-edge colorable graphs~\cite{chlebik2003approximation}.
For every vertex $x\in V(G)$, we have a project $p_x$ in $P$;
and for every edge $e \in E(G)$, we have two voters $w_e$ and $w_{\hat{e}}$ in $\calV$.
If $e=\{x,y\}$, then the voter $w_e$ approves projects $p_x$ and $w_{\hat{e}}$ approves  $p_y$, that is, $P_{w_e}=\{p_x\}$ and $P_{w_{\hat{e}}}=\{p_y\}$.
For every edge $e=\{x,y\}$, we have a set $\{p_x,p_y\}$ in the family $\calF$.
For every project $p\in P$, $c(p)=1$.
For every set $F\in \calF$, $b(F)=1$. We set $B=k$ and $u=3k$.
Since $G$ is $3$-edge colorable, we can partition the sets in $\calF$ in $3$ groups such that every two sets in a group are disjoint, and hence $\ell=3$.

Next, we prove the equivalence between the instance $(G,k)$ of \issmall\ and $(\calV,P,\calE,\calF,c,B,b,u)$ of \gpb.
In the forward direction, let $S$ be a solution to $(G,k)$.
We claim that $P_S=\{p_x\in P \colon x\in S\}$ is a solution to $(\calV,P,\calE,\calF,c,B,b,u)$.
Since $c(p)=1$, for every project $p$ in $P$, we have that $\sum_{p\in P_S}c(p)=k$.
Since $S$ is an independent set in $G$, for every $F \in \calF$, $\sum_{p\in P_S\cap F}c(p)=1$.
Since $G$ is a $3$-regular graph, every project in $P_S$ is approved by three voters.
Moreover, since every voter approves only one project,
$\sum_{v\in V}|P_v\cap P_S|=3k$.
In the backward direction, let $X$ be a solution to $(\calV,P,\calE,\calF,c,B,b,u)$.
We claim that $S=\{x\in V(G) \colon p_x\in X\}$ is a solution to $(G,k)$.
Since $b(F)=1$, for every $F\in \calF$, clearly, $S$ is an independent set in $G$.
We claim that $|S|=k$.
Clearly, $|S|\leq k$, otherwise $\sum_{p\in X}c(p) > k$, a contradiction.
If $|S|<k$, then as argued above, $\sum_{v\in V}|P_v\cap X|<3k$, a contradiction.

To get the other claim, in the above reduction, instead of adding two voters for every edge in the graph, we add only one voter who approves all the projects and set $u=k$.
The rest of construction remains the same.
Then, using the same arguments as above, we have the second claim.
\end{proof}

As \is\ on general graphs is $\W[1]$-hard wrt.\ the  solution size, the next result follows using similar arguments as used in the proof of Theorem~\ref{thm:paraNP-l+s+f}.

\begin{theorem}\label{thm:w-hard-app2}
\gpb is \W$[1]$-hard wrt.\ $u$ even when $s=2$, $\app=1$, and $b_{\max}=1$; and even when  $s=2$, $n=1$, and $b_{\max}=1$.
\end{theorem}

\begin{proof}
Here, we give a polynomial-time reduction from the \is\ problem, which is known to be $\W[1]$-hard wrt.\ $k$, where $k$ is a solution size~\cite{DowneyF95}.

For the first claim, the set of projects $P$, the cost function $c$, the family $\calF$, the function $b$, and the total budget $B$ is the same as in Theorem~\ref{thm:paraNP-l+s+f}.
Next, for every vertex $x \in V(G)$, we have a voter $w_x$ which approves only the project $p_x$.
Thus, $\app=1$.
We set $u=k$.
The proof of correctness is similar to the proof of Theorem~\ref{thm:paraNP-l+s+f}.

To get the other claim, we use exactly the same reduction as in the second part of the proof of Theorem~\ref{thm:paraNP-l+s+f}.
Note that we have also $u=k$, hence $\W[1]$-hardness wrt.\ $u$ follows.
\end{proof}

\section{Tractability of Hierarchical Instances}\label{sec:tractability1}

We start our quest for tractability by considering \gpb instances whose group structure is hierarchical; when~$\calF$ is hierarchical, we refer to the \gpb problem as \hpb.
Fortunately, such instances can be solved in polynomial time.

\begin{lemma}\label{lemma:hier-pb-in-p}
 \hpb can be solved in polynomial time.
\end{lemma}

\begin{proof}
 Let $(\calV,P,\calE,\calF,c,B,b,u)$ be a given instance of \hpb.
 W.l.o.g., assume that $P$ is a set in the family ${\cal F}$ (otherwise, add it to $\calF$ and set $b(P)=B$).
 Using Remark~\ref{rem:hierarchical-ld}, let $\cal{L}$ be an ordered layer decomposition of ${\cal F}$ such that every set is a subset of some set in the preceding layer, and note that the first layer is $\{P\}$.
 Let $S$ be a set in some layer, say $L_i$, such that $|S|>1$.
 Suppose that there exists a project $p\in S$ such that $p$ is not in any set of $L_{i+1}$. We add the set $\{p\}$ to $\calF$ and $L_{i+1}$, and set $b(\{p\})=c(p)$ (if the layer $L_{i+1}$ does not exist, then we add this new layer). Note that we might increase the number of layers by $1$. Let $\ell$ be the number of layers.

 Now, we solve the problem using dynamic programming.
 For a set $S \in \calF$ in the $i$-th layer, where $i\in [\ell-1]$, such that $|S|>1$, let $\mathtt{Part}_S$ denote the partition of a set $S$ such that every part in $\mathtt{Part}_S$ is a set in the $(i+1)$-th layer.
 For a singleton set $S$, $\mathtt{Part}_S=\{S\}$.
 For every set $S \in {\cal F}$, we order the parts in $\mathtt{Part}_S$ arbitrarily.
 Let $|\mathtt{Part}_S|$ denote the number of parts in $\mathtt{Part}_S$ and let $\mathtt{Part}_S(i)$ denote the $i$-th part in $\mathtt{Part}_S$.
 Our table entries are defined as follows.
 For every set
 $S \in \calF$,
 $j\in \lvert \mathtt{Part}_S \rvert$, and utility $z\in [u]$ we define $T[S,j,z]$ as the minimum cost of a bundle that has utility $z$ and that is a subset of projects in first $j$ parts in $\mathtt{Part}_S$ and all budget limits for $S$ and all $S' \in \calF, S' \subseteq S$ are satisfied;
 if the budgets limits cannot be satisfied whenever achieving utility $z$ by a subset of projects in first $j$ parts in $\mathtt{Part}_S$ then $T[S,j,z] = \infty$.
 
 For a project $p$, let $a(p)$ denote the number of voters who approve the project $p$ (approval score).
 For a set $S$, let $a(S) = \sum_{p\in S} a(p)$.
 We compute the table entries level-wise in bottom-up order, that is, we first compute the value corresponding to sets at lower levels.

{\bf Base case:} For every set $S$, where $S=\emptyset$ or $S\in {\calF}$ and $0\leq z \leq u$ we set
\begin{equation}\label{eq1:hpb-poly}
\begin{split}
        T[S,0,z] = \begin{cases}
        0 & \text{if } z=0 \\
        \infty & \text{otherwise}
        \end{cases}
\end{split}
\end{equation}

{\bf Recursive Step:} For every set $S\in {\calF}$, $j \in [|\mathtt{Part}_S|]$, and $0\leq z \leq u$, we compute $T[S,j,z]$ in two steps as follows:
\begin{align}
        \text{first step:}&\quad T[S,j,z] :=
        \min\limits_{0\leq z' \leq z} \{T[S,j-1,z-z'] +
         T[\mathtt{Part}_S(j),\lvert \mathtt{Part}_S(j) \rvert,z']\}\\
        \text{second step:}&\quad \text{if}\quad T[S,j,z] > b(S) \quad\text{then set}\quad T[S,j,z] := \infty.\label{eq3:hpb-poly}
\end{align}
The second step informs that the minimum cost bundle which achieve utility $z$ already exceeds the budget limit for set $S$.
Notice that if the budget was exceeded in $\mathtt{Part}_S(j)$ for a considered $z'$ then $\infty$ is propagated correctly.

Next, we prove that \eqref{eq1:hpb-poly} and \eqref{eq3:hpb-poly} correctly compute $T[S,j,z]$ for all sets $S\in {\calF}$, $0\leq j \leq \lvert \mathtt{Part}_S \rvert$, and $0\leq z \leq u$.
Note first that, for $j=0$, the subset of $S$ that contains projects only from the first $j$ parts is $\emptyset$, with cost $0$ and $a(\emptyset) = 0$, implying the correctness of the base case.

For \eqref{eq3:hpb-poly}, we provide inequalities in both directions.
In one direction, let $\hat{S}\subseteq S$ be a bundle of minimum cost that is a subset of projects in first $j$ parts in $\mathtt{Part}_S$, has utility $z$ and satisfies all budget constraints for $S' \in \calF, S' \subseteq S$ (for now we assume that such a bundle exists).
Note that cost of bundle $\hat{S}\cap \mathtt{Part}_S(j)$ is a valid entry for
$T[\mathtt{Part}_S(j),\lvert \mathtt{Part}_S(j) \rvert, a(\hat{S}\cap \mathtt{Part}_S(j))]$
and cost of the bundle $\hat{S}\setminus \mathtt{Part}_S(j)$ is a valid entry for
$T[S,j-1,z-a(\hat{S}\cap \mathtt{Part}_S(j))]$
(and they are not equal to $\infty$). Thus,
$$T[S,j,z] \geq
 \min\limits_{0\leq z' \leq z} \{T[S,j-1,z-z']+ T[\mathtt{Part}_S(j),\lvert \mathtt{Part}_S(j) \rvert,z']\}.$$
In the case, that every $\hat{S} \subseteq S$ which achieves utility $z$ cannot satisfy all budget constraints for $S' \in \calF, S' \subseteq S$, the inequality holds trivially.

For the other direction, for any $0\leq z' \leq z$, let $\hat{S} \subseteq S$ be a bundle of minimum cost which is a subset of projects in first $j-1$ parts in $\mathtt{Part}_S$, has utility $z-z'$ and satisfies all budget constraints for $S' \in \calF, S' \subseteq \hat{S}$ (for now we assume that such a bundle exists).
Additionally, let $\tilde{S} \subseteq \mathtt{Part}_S(j)$ be a bundle of minimum cost with utility $z'$ which satisfies all budget constraints for $S' \in \calF, S' \subseteq \mathtt{Part}_S(j)$.
Note that $\hat{S}\cap \tilde{S}=\emptyset$.
Thus, $\hat{S}\cup \tilde{S}$ is a valid entry for $T[S,j,z]$; so
\begin{equation*}
\begin{split}
        T[S,j,z] \leq
        \min\limits_{0\leq z' \leq z} \{T[S,j-1,z-z']+
         T[\mathtt{Part}_S(j),\lvert \mathtt{Part}_S(j) \rvert,z']\}.
\end{split}
\end{equation*}
In the case that such $\hat{S}$ or $\tilde{S}$ do not exist for some $0\leq z' \leq z$, then the inequality still holds as some elements of a set over the minimum operator are equal to $\infty$.
This proves the correctness of (\ref{eq3:hpb-poly}).

Using (\ref{eq1:hpb-poly}) and (\ref{eq3:hpb-poly}), we compute $T[S,j,z]$, for all $S\in {\calF}$, $0\leq j \leq \lvert \mathtt{Part}_S \rvert$, and $0\leq z\leq u$ level-wise in bottom-up order.
Note that $T[P,|\mathtt{Part}_P|,u]$ implies the answer, i.e., the algorithm outputs a response YES if $T[P,|\mathtt{Part}_P|,u] \leq B$, and NO otherwise.
The running time of the algorithm is $\Oh(|\calF|su)$ as the size of any set in ${\calF}$ is at most~$s$.
This can be upperbounded by $\Oh(|\calF|nm^2)$ which is polynomial in the input size.
\end{proof}

Some instances might not be hierarchical but only close to being such, thus we study two distance parameters, namely, the minimum number $D_g$ of groups and the minimum number $D_p$ of projects, respectively, whose deletion results in a hierarchical instance.
Indeed, having efficient algorithms for such instances means that even more instances can be efficiently solved (e.g., instances with group structures corresponding to thematic districts in which some projects fit several groups).

We have the following lemmas---used later in the proofs for the parameterized complexity of \gpb wrt.\ $D_g+s$ and $D_p$---which are concerned with computing the set of groups/projects whose deletion leads to hierarchical instance.
Their proofs follow branching arguments, as, for $D_g$, at least one set from each pair of conflicting groups shall be deleted, and, for $D_p$, for a pair of conflicting groups $G_1, G_2$, either $G_1 \setminus G_2$ or $G_2 \setminus G_1$ or $G_1 \cap G_2$ shall be removed.

\begin{lemma}\label{lem:D_g deletion set}
 There exists an algorithm, running in $\Ohstar(2^{D_g})$ time, that finds a minimum-sized set of groups whose deletion results in a hierarchical instance.
\end{lemma}
\begin{proof}
 First, observe that, if there are two non-disjoint groups $G_1,G_2$ such that neither $G_1 \subseteq G_2$ nor $G_2 \subseteq G_1$, then we have to delete one of this group.
 So, we branch as follows:
 delete either $G_1$ or $G_2$, and decrease $D_g$ by $1$ in each branch.
 We only branch further if $D_g >0$.
 If there is a leaf in this branching tree for which the reduced family is hierarchical, then this leaf gives us required deletion set.
 The algorithm runs in $\Ohstar(2^{D_g})$-time.
\end{proof}

\begin{lemma}\label{lem:D_p deletion set}
 There exists an algorithm, running in $\Ohstar(3^{D_p})$ time, that finds a minimum-sized set of projects whose deletion results in a hierarchical instance.
\end{lemma}
\begin{proof}
 First, observe that, if there are two non-disjoint groups $G_1,G_2$ such that neither $G_1 \subseteq G_2$ nor $G_2 \subseteq G_1$, then we have to delete either intersecting part so that both the groups are disjoint or we delete $G_1\setminus G_2$ or $G_2 \setminus G_1$ so that one of them is subset of other.
 So, we branch as follows: Delete $G_1\cap G_2$ and decrease $D_p$ by $|G_1 \cap G_2|$ or delete $G_1\setminus G_2$ and decrease $D_p$ by $|G_1\setminus G_2|$ or delete $G_2 \setminus G_1$ and decrease $D_p$ by $|G_2\setminus G_1|$.
 Clearly, in each branch, we decrease $D_p$ by at least one.
 If there is a leaf in this branching tree for which the reduced family is hierarchical, then this leaf gives us required deletion set.
 The algorithm runs in $\Ohstar(3^{D_p})$-time.
\end{proof}

Unfortunately, we have the following intractability result.

\begin{theorem}\label{thm:partition-pb-xp-dg}
 \gpb is $\NP$-hard even when $D_g=2$ and $\ell=2$.
\end{theorem}

\begin{proof}
 We provide a polynomial-time reduction from the {\sc Partition} problem, in which, given a set of natural numbers $X=\{x_1,\ldots,x_{\tilde{n}}\}$, the goal is to partition $X$ into two parts, $X_1$ and $X_2$, such that the sum of the numbers in $X_1$ is same as the sum of the numbers in $X_2$.
 The {\sc Partition} problem is known to be $\NP$-hard~\cite{GareyJ79}.
 We construct an instance $(\calV,P,\calE,\calF,c,B,b,u)$ of \gpb as follows:
 for every integer $x_i\in X$, we add two projects in $P$, say $x_i^1$ and $x_i^2$, and set $c(x_i^1)=c(x_i^2)=x_i$.
 We define only one voter $v$ approving all the projects.
 For every $x_i \in X$, we add $\{x_i^1,x_i^2\}$ to $\calF$ and define $b(\{x_i^1,x_i^2\})=x_i$.
 Furthermore, we add sets $\{x_i^1\colon x_i\in X\}$ and $\{x_i^2\colon x_i\in X\}$ to $\calF$; and set $b(\{x_i^1\colon x_i\in X\})= b(\{x_i^2\colon x_i\in X\})=\sum_{x_i\in X} \nicefrac{x_i}{2}$.
 Let $B=\sum_{x_i\in X}x_i$ and $u=\tilde{n}$.
 Note that the layerwidth of $\calF$ is $2$ and by removing the sets $\{x_i^1\colon x_i\in X\}$ and $\{x_i^2\colon x_i\in X\}$ from $\calF$ leads to a hierarchical family.

 The main idea of the correctness is that the decision whether to fund project $x_i^1$ or project $x_i^2$ (as both cannot be funded) corresponds to deciding whether to take $x_i$ to the first part or the second part of the solution to the instance of {\sc Partition}.

 We next prove formally the equivalence between the instance $X$ of the {\sc Partition} problem and the instance  $(\calV,P,\calE,\calF,c,B,b,u)$ of \gpb.

 In the forward direction, let $(X_1,X_2)$ be a solution to $X$.
 We create a subset of projects, $S\subseteq X$, as follows.
 If $x_i\in X_1$, where $i\in [\tilde{n}]$, then we add $x_i^1$ to $S$, otherwise we add $x_i^2$ to $S$.
 Since $\lvert S \rvert = \tilde{n}$, clearly, the utility of $S$ is $\tilde{n}$.
 Since we pick either $x_i^1$ or $x_i^2$ in $S$, for every $i\in [\tilde{n}]$, for every set $F=\{x_i^1,x_i^2\}\in \calF$, $\sum_{p\in F\cap S} = x_i$.
 Since the sum of all the numbers in $X_1$ is  $\nicefrac{\sum_{x_i\in X}x_i}{2}$, for the set $F=\{x_i^1\colon x_i\in X\} \in \calF$, $\sum_{p\in F\cap S}= \nicefrac{\sum_{x_i\in X}x_i}{2}$.
 Similarly, for the set  $F=\{x_i^2\colon x_i\in X\} \in \calF$, $\sum_{p\in F\cap S}= \nicefrac{\sum_{x_i\in X}x_i}{2}$.

 In the backward direction, let $S$ be a solution to $(\calV,P,\calE,\calF,c,B,b,u)$.
 Since $u\geq \tilde{n}$, clearly, at least $\tilde{n}$ many projects belong to $S$.
 Note that for any $i\in [\tilde{n}]$, $x_i^1$ and $x_i^2$ both cannot belong to $S$ as $c(x_i^1)=c(x_i^2)=x_i$ and $b(\{x_i^1,x_i^2\})=x_i$.
 Thus, as there are only $2\tilde{n}$ many projects, either $x_i^1$ or $x_i^2$ belong to $S$, but not both.
 We construct a partition of $X=(X_1,X_2)$ as follows:
 if $x_i^1 \in S$, where $i\in [\tilde{n}]$, then add $x_i$ to $X_1$, otherwise add $x_i$ to $X_2$.
 Clearly, the sum of the numbers in $X_1$ is at most $\nicefrac{\sum_{x_i\in X}x_i}{2}$ as $b(\{x_i^1\colon x_i\in X\})=\nicefrac{\sum_{x_i\in X}x_i}{2}$.
 Similarly, the sum of the numbers in $X_2$ is at most $\nicefrac{\sum_{x_i\in X}x_i}{2}$.
 Now, suppose the sum of the numbers in at least one of these sets is less than $\nicefrac{\sum_{x_i\in X}x_i}{2}$.
 This implies that the total sum of the numbers in $X_1$ or $X_2$ is less than $\sum_{x_i\in X}x_i$, this contradicts that for every $i\in [\tilde{n}]$, either $x_i^1$ or $x_i^2$ is in $S$.
 This completes the proof.
\end{proof}

Nevertheless, combining with $s$ helps.

\begin{theorem}\label{thm:partition-pb-fpt-dg+s}
 \gpb is $\FPT$ wrt.\ $D_g+s$.
\end{theorem}
\begin{proof}
 We can try all possible subsets of the deletion set;
 then, for each deletion set, we can try all possibilities of which projects to fund (i.e., going over all subsets of projects in the union of the groups in the deletion set); then, for the remaining budget, we solve using the polynomial-time algorithm for \hpb, as described in the proof of Lemma~\ref{lemma:hier-pb-in-p}.
 The running time of the algorithm is $\Ohstar(2^{D_g \cdot s})$.
\end{proof}

In contrast to Theorem \ref{thm:partition-pb-xp-dg}, parametrization by the delete-project-distance of an instance to be hierarchical is tractable.

\begin{theorem}\label{thm:partition-pb-fpt-d}
 \gpb is $\FPT$ wrt.\ $D_p$.
\end{theorem}

\begin{proof}
 We can use Lemma~\ref{lem:D_p deletion set} to find the deletion set.
 Then we can go over all possibilities of which projects to fund from the deletion set and solve the remaining hierarchical instance, after updating the budgets, using the polynomial-time algorithm of Lemma~\ref{lemma:hier-pb-in-p}.
\end{proof}

\section{Tractability with Few Groups}\label{sec:tractability2}

Next we concentrate on the number $g$ of groups as a parameter as, indeed, the groups are the new ingredient we bring to the standard model of PB.
Practically, the number $g$ of groups may be set by the entity organizing the PB process, thus can be as small as the organizer wishes.
First, we have the following result for the parameter $g$.

\begin{theorem}\label{thm:xp-g}
\gpb is $\XP$ wrt.\ $g$.
\end{theorem}
\begin{proof}
The algorithm follows from an enumeration over project types $t_R$, $R \subseteq \calF$, where $t_R$ is defined as follows:
  a project of type $t_R$ belongs to all groups in $R$ and to none of the groups in $\calF \setminus R$.
  We override the notation and write $p \in t_R$ to say that project $p \in P$ has type $t_R$.
  Note that each project has exactly one type out of the total of $2^g$ types of projects.
  For every project type $t_R$, and for every value $\mu \in \{0, 1, 2, \ldots, u\}$, we use standard dynamic programming to find the cheapest set of projects of type $t_R$ that achieves at least utility $\mu$.
  (For each of the $2^g$ types, this takes polynomial-time.)
  Then, we go over all possibilities of achieving some utility from each of the types.
  As we have $u+1$ possibilities for each project type, it follows that, in total, we consider at most $(u+1)^{2^g}$ cases.
  If some solution is feasible then the instance is a yes-instance, otherwise it is a no-instance.

  For correctness we focus only on yes-instances (for no-instances the correctness is straightforward).
  Fix some feasible solution $X^*$ to the instance.
  It achieves the desired total utility $u$ by getting some utility $u^*(R) \in \mathbb{N}$ from projects of type $t_R$.
  Note that our algorithm considers a solution $X$ constructed in a case when it receives utility at least $\min\{u^*(R),u\}$ from projects of type $t_R$.
  We will argue that this solution is feasible. $X$ achieves utility of at least $\sum_{R \subseteq \calF} \min\{u^*(R),u\} \geq u$, where the inequality holds because of the following:
  (1) if there exists some $R \subseteq \calF$ such that $u^*(R) \geq u$,
  then this is straightforward;
  (2) if for all $R \subseteq \calF$ we have that $u^*(R) < u$,
  then this follows from the feasibility of $X^*$,
  i.e., $\sum_{R \subseteq \calF} u^*(R) \geq u$.
  For group budget limits we note that $X$ is constructed from min-cost bundles of projects of type $t_R$ that achieve utility of at least $\min\{u^*(R),u\}$.
  It means that the cost of projects of type $t_R$ in $X$ cannot be more expensive than the cost of projects of type $t_R$ in $X^*$.
  By aggregating this cost for each group $F \in \calF$ separately we get that
  $\sum_{p \in F \cap X} c(p) = \sum_{R: F \in R \subseteq \calF}\sum_{p \in t_R \cap X} c(p) \leq \sum_{R: F \in R \subseteq \calF}\sum_{p \in t_R \cap X^*} c(p) = \sum_{p \in F \cap X^*} c(p) \leq b(F)$.
  Analogously we show feasibility regarding a global budget.

  Since the running time of the algorithm above is $\Ohstar((u+1)^{2^g})$, we finish the proof.
\end{proof}

Unfortunately, we do not know whether \gpb is $\FPT$ wrt.\ $g$; indeed, this is the main question left open.

\begin{openquestion}\label{thm:fpt-g}
Is \gpb $\FPT$ wrt.\ $g$?
\end{openquestion}

Note, however, that in Subsection~\ref{sec:w-hardness-utility} we provide a $\W$-hardness wrt.\ $g$ proof, albeit for a slightly more general problem, in which we are  also given utility requirements for each group.

Next we consider combined parameters.

\begin{theorem}\label{thm:fpt-u+g}
 \gpb is $\FPT$ wrt.\ $g+u$.
\end{theorem}
\begin{proof}
 Follows from the proof of Theorem~\ref{thm:xp-g} as the running time of the algorithm is $\Ohstar((u+1)^{2^g})$.
\end{proof}

Careful Mixed Integer Linear Programming (MILP) formulation implies the following.

\begin{theorem}\label{thm:fpt-g+n}
\gpb is $\FPT$ wrt.\ $g+n$.
\end{theorem}

\begin{proof}
Recall that w.l.o.g. we assume that every project is approved by at least one voter.
We construct an MILP in the following way.
We define a type of a project by a pair $(R,w)$, where $R \subseteq \calF$ and $w \in [n]$.
A project of type $(R,w)$ belongs to all the groups in $R$ (and to none of the groups in $\calF \setminus R$) and it is approved by exactly $w$ voters.
Note that we have $n \cdot 2^g$ types of projects.

We define an integer variable $x_{R,w}$ for all $R \subseteq \calF$ and $ w \in [n]$, meaning how many projects of type $(R,w)$ are in a solution.
Let $|(R,w)|$ be the number of projects of type $(R,w)$.

Next, we split $x_{R,w}$ into a sum of $|(R,w)|$ real variables:
\begin{equation}\label{milp:xtoy}
  x_{R,w} = \sum_{i \in [|(R,w)|]} y_{R,w,i},
\end{equation}
where $y_{R,w,i} \in [0,1]$ is a continuous extension of a binary variable that indicates whether we take the $i$-th cheapest projects of type $(R,w)$ to a solution.
From equation~\eqref{milp:xtoy} we get also $x_{R,w} \in \{0,1,\dots, |(R,w)|\}$.
Note that we have $\sum_{R \subseteq \calF} \sum_{w \in [n]} |(R,w)| = m$ real variables $y_{R,w,i}$ because each project has exactly one type.

We need to implement the budget function.
We write a constraint for each group $F \in \calF$:
\begin{equation}\label{milp:groupbudgets}
  \sum_{R: F \in R \subseteq \calF} \sum_{w \in [n]} \sum_{i \in [|(R,w)|]} y_{R,w,i} \cdot c(R,w,i) \leq b(F),
\end{equation}
where $c(R,w,i)$ is the cost of the $i$-th cheapest project of type $(R,w)$.
Similarly, we add a global budget limit constraint as follows:
\begin{equation}\label{milp:globalbudget}
  \sum_{R \subseteq \calF} \sum_{w \in [n]} \sum_{i \in [|(R,w)|]} y_{R,w,i} \cdot c(R,w,i) \leq B.
\end{equation}

The remaining ingredient of MILP is the objective function:
\begin{equation}\label{milp:objective}
  \max \sum_{R \subseteq \calF} \sum_{w \in [n]} w \cdot x_{R,w}.
\end{equation}

We can transform any optimal solution $(x^*,y^*)$ of the MILP into an optimal solution $(x^*,y^{\integer})$ consisting of integer variables only.
In particular, we define $y^{\integer}_{R,w,i} = 1$ for $i \in \{1,\dots, x^*_{R,w}\}$ and $y^{\integer}_{R,w,i} = 0$ for $i \in \{x^*_{R,w}+1, \dots, |(R,w)|\}$.

\paragraph{Optimality.} The objective value of such a new solution $(x^*,y^{\integer})$ is the same as for $(x^*,y^*)$ (hence optimal) because the objective function depends only on $x_{R,w}$ and both the solutions have the same value for variables $x_{R,w}$.

\paragraph{Feasibility.} We have
\begin{align*}
    \sum_{i \in [|(R,w)|]} y^{\integer}_{R,w,i}
  = \sum_{i = 1}^{x^*_{R,w}} y^{\integer}_{R,w,i} + \sum_{i = x^*_{R,w}+1}^{|(R,w)|} y^{\integer}_{R,w,i}
  = \sum_{i = 1}^{x^*_{R,w}} 1                    + \sum_{i = x^*_{R,w}+1}^{|(R,w)|} 0
  = x^*_{R,w}.
\end{align*}
hence the integer solution $(x^*,y^{\integer})$ is feasible for equation~\eqref{milp:xtoy}.

We have $c(R,w,i) \leq c(R,w,i+1)$ and $\sum_{i \in [|(R,w)|]} y^{\integer}_{R,w,i} = x^*_{R,w} = \sum_{i \in [|(R,w)|]} y^*_{R,w,i}$,
hence we get
$$\sum_{i \in [|(R,w)|]} y^{\integer}_{R,w,i} \cdot c(R,w,i) \leq \sum_{i \in [|(R,w)|]} y^*_{R,w,i} \cdot c(R,w,i).$$
From this and feasibility of $(x^*,y^*)$ we get that the integer solution $(x^*,y^{\integer})$ is also feasible for equations~\eqref{milp:groupbudgets} and \eqref{milp:globalbudget}.

\paragraph{Running Time.}
The MILP has at most
$n \cdot 2^g$ integer variables,
$m$ real variables and $n \cdot 2^g +2m+g+1$ constraints
($n \cdot 2^g$ many constraints of type~\eqref{milp:xtoy}, $2m$ many constraints for variables $y_{R,w,i}$, $g$ many group budget constraints~\eqref{milp:groupbudgets} and $1$ global budget constraint~\eqref{milp:globalbudget}).
We can solve MILP in time $\FPT$ wrt. $p$,
where $p$ is the number of integer variables~\cite{lenstra1983integer}.
In particular, our MILP can be solved in $p^{\Oh(p)} \cdot |I|^{\Oh(1)} \leq \Ohstar((n \cdot 2^g)^{\Oh(n \cdot 2^g)}) \leq \Ohstar(2^{n \log(n) \cdot 2^{\Oh(g)}})$ time~\cite{Kannan87,BredereckFNST20}.
\end{proof}

\begin{remark}
Note that $n \cdot 2^g$ is only an upper-bound for the number of integer variables in the MILP.
Indeed, we can write a more strict upper-bound.
Let $t$ be the number of non-empty types $(R,w)$, i.e.,
$t = |\{(R,w): R \subseteq \calF, w \in [n], |(R,w)|>0 \}|$.
Let $A$ be the maximum approval score over projects, i.e.,
$A = \max_{p \in P} |\{ v \in \calV: p \in P_v \}|$.
Constraint~\eqref{milp:xtoy} sets $x_{R,w}$ to be $0$ for empty types $(R,w)$.
It means that we have at most $A \cdot t$ free integer variables, hence \gpb is $\FPT$ wrt.\ $A \cdot t$.
\end{remark}

Also combining $g$ with the budget $B$ helps.

\begin{theorem}\label{thm:fpt-g+b}
 \gpb is $\FPT$ wrt.\ $g+B$.
\end{theorem}

\begin{proof}
 Recall that w.l.o.g. we assumed that $b(F) \leq B$ for every $F \in \calF$, hence $b_{\max} \leq B$.
 We can apply the dynamic programming algorithm for \xdk{(g+1)}, which runs in time upper-bounded by $\Ohstar{(n \cdot (b_{\max}+1)^g (B+1))} \leq \Ohstar((B+1)^{(g+1)})$ \cite[Section 9.3.2]{knapsackbook}.
\end{proof}

\begin{remark}
We mention that the running time claimed by Kellerer \textit{et al.}~\cite[Section 9.3.2]{knapsackbook} is $\Ohstar{(n \cdot b_{\max}^g \cdot B)}$.
This would give a polynomial time algorithm for hard instances after the reduction in Theorem~\ref{thm:paraNP-l+s+f} because there we have $b_{\max}=1$.
It would imply $\P=\NP$.
The logical flaw of Kellerer \textit{et al.}~\cite{knapsackbook} comes from counting the size of the DP table, in particular,
in every dimension we should count the number of sub-problems by considering all possible values for a budget, i.e.,
all integer numbers between $0$ and the budget limit in a dimension.
Indeed, asymptotically this $+1$ difference usually does not matter; for us, however, it is crucial.
\end{remark}

We observe that the running time in the proof of Theorem~\ref{thm:fpt-g+b} is $\Ohstar{((b_{\max}+1)^g (B+1))}$, which is not $\FPT$ wrt.\ $g+b_{\max}$, since $B$ may be non-polynomial in the input size.
In light of Open Question~\ref{thm:fpt-g}, analyzing the parameterized complexity of \gpb wrt.\ $g+b_{\max}$ could serve as a valuable intermediate step.

\subsection{FPT Approximation Scheme for g}

Recall that \gpb is $\XP$ wrt.\ $g$ (Theorem~\ref{thm:xp-g}) and recall our open question regarding whether \gpb is $\FPT$ wrt.\ $g$ (Open Question~\ref{thm:fpt-g}).
Next we show an approximation scheme for \mgpb that is $\FPT$ wrt.\ $g$
(compare this result also to that described later, in Theorem~\ref{thm:sqrt-g-apx},
showing that there does not exist a constant-factor approximation algorithm unless $\P=\NP$, even if $g$ is as small as $m^2$).

In particular, our approximation notion is the following:
an algorithm has an approximation factor $\alpha \geq 1$ if it always outputs a solution that has at most $\alpha$ factor less utility than the optimal solution.

\begin{theorem}\label{thm:fpt-as-simpler}
  There exists an algorithm that for any fixed $\epsilon>0$ finds an $(1+\epsilon)$-approximate solution to \mgpb in $\FPT$ time wrt.~$g$.
\end{theorem}

\begin{proof}
 The idea of the algorithm is as follows.
 First, we reduce the given instance of \gpb to an instance of \gpb with an additional feasibility restriction, in particular,
 such that a feasible solution has to contain exactly one project from each project type,
 where a type of a project is uniquely defined by the family of groups to which the project belongs.
 The reduction, shown below, takes $\FPT$ time wrt.\ $g$.
 In the second step we will round down the approval score of each project to the closest multiplicity of $(1+\epsilon)$, in effect bounding the number of different approval scores of a project to the logarithmic function of the input size.
 Then we will apply a brute-force enumeration that runs in $\FPT$ time wrt.\ $g$.

 More formally, let us fix $\epsilon>0$ and an instance $\calI = (\calV,P,\calE,\calF,c,B,b)$ of \mgpb.
 Recall that w.l.o.g. we assumed that each project is approved by at least one voter.
 Let $a\colon P \rightarrow \{1,2,\dots,|\calV|\}$ be an approval score function, i.e.,
 $a(p) = |\{ v \in \calV: p \in P_v \}|$.,
 Notice that the approval score function $a(\cdot)$ can be encoded in unary
 (instead of having voters explicitly).
 Let $A$ be the total approval score of all the projects, e.g., $A = \sum_{p \in P} a(p)$.
 Let $u^*(\calI)$ be the value (total utility) of an optimal solution to $\calI$.
 To avoid triviality, w.l.o.g. we assume $u^*(\calI)>0$.

 Now, given a subfamily $R \subseteq \calF$, we say that a project $p$ is of type $R$ if it belongs to all the groups in $R$ and to none of the groups in $\calF \setminus R$ (so every project has an unique type).
 We have at most $2^g$ types of projects.

 First, we fix an optimal solution $X^*$ and we do the following preprocessing on the instance $\calI$.
 For every project type, we guess whether at least one project of the type is contained in $X^*$, and if none, we delete all projects of that type.
 We can do this preprocessing in $\Ohstar(2^{2^g})$ time.
 Note that, after this step, the number of project types cannot increase because the number of projects cannot increase.
 (As this is just a preprocessing, in our next steps we override the notation and use $P$ for the set of projects after the preprocessing.)
 We define the number of project types after the preprocessing as $t$, with $t \leq 2^g$.

 For every project type $R \subseteq \calF$ we run a dynamic programming procedure that outputs the following:
 For every value $v \in \{1,2,\dots,A\}$, we compute $\cost(R,v)$, which is the minimum cost of a subset of projects of type $R$ whose total value is exactly $v$ (it is equal to $\infty$ if there is no such subset).
 Also we store a bundle of projects, $\bundle(R,v)$, that realizes the minimum cost $\cost(R,v)$.
 We can compute $\cost(R,v)$ together with $\bundle(R,v)$ in time upper-bounded by $\Ohstar(t \cdot A \cdot |P|) = \Ohstar(2^g)$.

 Now, we create a new instance $\calI' = (\calV',P',\calE',\calF',c',B',b')$ of \gpb as follows.
 For every project type $R \subseteq \calF$ and every value $v \in \{1,2,\dots,A\}$ such that $\cost(R,v)$ is not $\infty$, we define a project $\proj(R,v) \in P'$ of cost $c'(\proj(R,v)) = \cost(R,v)$ and approval score $a'(\proj(R,v))$ equals to $v$
 (equivalently we can define $A$ many voters in $\calV'$, where the $i$-th voter approves all the projects $\proj(R,v)$ such that $v \geq i$).
 We define the type of all the projects $\proj(R,v)$, $v \in \{1,2,\dots,A\}$, as $R$.
 Note that a project $\proj(R,v)$ corresponds to a bundle of projects of type $R$ from the original instance.

 We keep the same global budget limit, i.e., $B' = B$.
 For every group $F \in \calF$, we define a group $T_F \in \calF'$ that contains all projects $\proj(R,v)$ whose type $R$ contains the group $F$, i.e., $T_F = \{\proj(R,v): F \in R\}$.
 We define $b'(T_F) = b(F)$.

 We show correspondence of feasible solutions in both instances.
 Let $\calI'_1$ be the instance $\calI'$ of \gpb restricted to solutions containing exactly one project of each type.
 \begin{lemma}\label{lemma:fpt-as-trans-to-origin-simpler}
  Every feasible solution to $\calI'_1$ can be transformed into a feasible solution to $\calI$ with the same utility.
 \end{lemma}

 \begin{proof}
  Let $X'$ be a feasible solution to $\calI'_1$.
  We construct a solution $X$ to $\calI$ as follows.
  For every $\proj(R,v) \in X'$ we add to $X$ projects stored in $\bundle(R,v)$.

  First, we show that $X$ is feasible to $\calI$.
  For warm-up we show that $X$ keeps a global budget $B$.
  We have the following chain of (in)equalities:
  \begin{align*}
   \sum_{p \in X} c(p) = \sum_{\proj(R,v) \in X'} \sum_{p \in \bundle(R,v)} c(p)
   = \sum_{\proj(R,v) \in X'} \cost(R,v)
   = \sum_{\proj(R,v) \in X'} c'(\proj(R,v)) \leq B' = B,
  \end{align*}
  where the first equality follows from definition of $X$, the second equality follows from definition of $\cost(R,v)$,
  the third equality follows from definition of $c'$, and the inequality follows from feasibility of $X'$ to $\calI'_1$.

  Next, we show that $X$ keeps budgets $b(\calF)$ for every $F \subseteq \calF$.
  The proof of this fact is analogous to the proof of feasibility for a global budget $B$,
  but we limit the first sum to projects of type $R$ that contains group $F$.
  We have
  \begin{align*}
   \sum_{p \in F \cap X} c(p) = \sum_{R \subseteq \calF: F \in R}\sum_{\proj(R,v) \in X'} \sum_{p \in \bundle(R,v)} c(p)
   = &\sum_{R \subseteq \calF: F \in R}\sum_{\proj(R,v) \in X'} \cost(R,v)\\
   = &\sum_{R \subseteq \calF: F \in R}\sum_{\proj(R,v) \in X'} c'(\proj(R,v))
   \leq b'(T_F) = b(F),
  \end{align*}
  where the inequality follows from feasibility of $X'$ to $\calI'$.
  It finishes the proof of feasibility of $X$ to $\calI$.
  The last one step is to show that the utility achieved by $X$ is equal to the utility achieved by $X'$. Indeed we have
  \begin{align*}
   \sum_{p \in X} a(p)
   = \sum_{\proj(R,v) \in X'} \sum_{p \in \bundle(R,v)}\hspace{-10pt}a(p)
   = \sum_{\proj(R,v) \in X'} \hspace{-10pt} v
   = \sum_{p' \in X'} a'(p'),
  \end{align*}
  where the second equality follows from definition of $X$,
  the third equality follows from definition of $\bundle(R,v)$,
  the forth equality follows from definition of $a'$.
 \end{proof}

 \begin{lemma}\label{lemma:fpt-as-utility-correspond-simpler}
  We have $u^*(\calI) \leq u^*(\calI'_1)$.
 \end{lemma}
 
 \begin{proof}
  Recall that $X^*$ is a fixed optimal solution to $\calI$.
  We will construct a feasible solution $X'$ to $\calI'_1$.
  Intuitively, for every type $R \subseteq \calF$,
  we take projects of type $R$ from $X^*$,
  we count their total utility and
  we put to $X'$ a min-cost bundle of type $R$ that achieve the same utility.
  Formally, for every project type $R \subseteq \calF$,
  a solution $X'$ contains project $\proj(R,\sum_{p \in X^* \cap R} a(p))$,
  where $p \in R$ denotes that $p$ has type $R$.
  Obviously, $X'$ contains exactly one project of each type.
  Next we show that $X'$ keeps a global and all group budgets.
  Cost of solution $X'$ is equal to
  \begin{align*}
   \sum_{p' \in X'} c'(p')
   = \sum_{R \subseteq \calF} c'(\proj(R,\sum_{p \in X^* \cap R} a(p)))
   = \sum_{R \subseteq \calF} \cost(R,\sum_{p \in X^* \cap R} a(p))
   \leq \sum_{R \subseteq \calF} \sum_{p \in X^* \cap R} c(p)
   =  \sum_{p \in X^*} c(p)
   \leq B = B',
  \end{align*}
  where the first equality follows from definition of $X'$,
  the second equality follows from definition of $c'$,
  the first inequality follows from definition of $\cost(R,v)$ being a proper min-cost bundle,
  and the second inequality follows from feasibility of $X^*$ to $\calI$.

  Next, we show that $X'$ keeps group budgets $b'(T_F)$ for every $F \subseteq \calF$.
  The proof of this fact is analogous to the proof of feasibility for a global budget $B$.
  The cost of projects taken from $T_F$ is equal to
  \begin{align*}
   \sum_{p' \in T_F \cap X'} c'(p')
   = \sum_{R \subseteq \calF: F \in R} c'(\proj(R,\sum_{p \in X^* \cap R} a(p)))
   = &\sum_{R \subseteq \calF: F \in R} \hspace{-10pt} \cost(R,\sum_{p \in X^* \cap R} a(p))\\
   \leq &\sum_{R \subseteq \calF: F \in R} \sum_{p \in X^* \cap R} \hspace{-7pt} c(p)
   = \sum_{p \in F \cap X^*} \hspace{-5pt} c(p)
   \leq b(F) = b'(T_F).
  \end{align*}
  It finishes the proof of feasibility of $X'$ to $\calI'_1$.

  The last one step is to show that the utility achieved by $X'$ is equal to
  \begin{align*}
   \sum_{p' \in X'} a'(p')
   = \sum_{R \subseteq \calF} a'(\proj(R,\sum_{p \in X^* \cap R} a(p)))
   = \sum_{R \subseteq \calF} \sum_{p \in X^* \cap R} a(p)
   = \sum_{p \in X^*} a(p)
   = u^*(\calI),
  \end{align*}
  where the first equality follows from definition of $X'$,
  the second equality follows from definition of $a'$,
  and the last equality follows from optimality of $X^*$.
  As $X'$ is a feasible solution to $\calI'_1$ hence $u^*(
  \calI'_1) \geq u^*(\calI)$ as required.
 \end{proof}

 Note that from Lemma~\ref{lemma:fpt-as-trans-to-origin-simpler} and Lemma~\ref{lemma:fpt-as-utility-correspond-simpler} we could derive $u^*(\calI) = u^*(\calI'_1)$.
 Although the crucial ingredient of Lemma~\ref{lemma:fpt-as-trans-to-origin-simpler} is that we can transform a feasible solution to $\calI'_1$ into a feasible solution to $\calI$ (keeping the value of the solution) in polynomial-time.

 In the second step (called \emph{bucketing}) we round down the approval score of each project to the closest multiple of $(1+\epsilon)$.
 Let $\calI''_1$ be an instance after the bucketing procedure
 (note that we do not change costs and budget limits when bucketing the approval scores).

 \begin{lemma}\label{lemma:fpt-as-apx-lose-simpler}
  We have $u^*(\calI''_1) \geq \frac{u^*(\calI'_1)}{1+\epsilon}$.
 \end{lemma}
 
 \begin{proof}
  Let $a''$ be an approval score function in instance $\calI''_1$ and let $X'$ be an optimal solution to $\calI'_1$.
  A set $X'$ is a feasible solution to $\calI''_1$ as well.
  Therefore we have
  \begin{align*}
   u^*(\calI''_1)
   \geq \sum_{p' \in X'} a''(p')
   \geq \sum_{p' \in X'}\frac{a'(p')}{1+\epsilon}
   = \frac{u^*(\calI'_1)}{1+\epsilon},
  \end{align*}
   where the first inequality follows from feasibility of $X'$ to instance $\calI''_1$,
   the second inequality follows from definition of bucketing,
   and the equality follows from optimality of $X'$ to instance $\calI'_1$.
 \end{proof}

 Because of bucketing, it is possible that two projects of the same type but with different approval scores are rounded down to the same value.
 Hence we keep the project of minimum cost.
 So, overall, after bucketing we have at most $\log_{1+\epsilon}(A)$ projects of each type. For each project type, we branch on which project of that type is selected and we store the best solution $X''$ (it has utility $u^*(\calI''_1)$).
 This takes $\Ohstar((\log_{1+\epsilon}(A))^{t})$ time, which can be bounded as follows.

 \begin{lemma}\label{lem:fpt-as-time-simpler}
  For a fixed $\epsilon>0$ we have
  $$\Ohstar((\log_{1+\epsilon}(A))^{t}) \leq \Ohstar(2^{\frac{1}{2} \cdot 4^g}).$$
 \end{lemma}

 \begin{proof}
  The left-hand side is equal to
  $\Ohstar(2^{2^g \cdot \log\log_{1+\epsilon}(A)})$
  and, using $xy \leq (x^2+y^2)/2$, it is not bigger than
  $\Ohstar(2^{\frac{1}{2} \cdot 4^g + \frac{1}{2} \cdot \log^2\log_{1+\epsilon}(A)}) \leq \Ohstar(2^{\frac{1}{2} \cdot 4^g + \delta(\epsilon) \cdot \log(A)})$,
  where $\delta(\epsilon)$ is a constant dependent on fixed $\epsilon$.
  Hence, the expression is bounded by
  $\Ohstar(2^{\frac{1}{2} \cdot 4^g}) \cdot \Ohstar(A^{\delta(\epsilon)}) = \Ohstar(2^{\frac{1}{2} \cdot 4^g})$
  because $A \leq |\calV||P|$.
 \end{proof}

 It shows the algorithm runs in $\FPT$ time wrt. $g$.
 Note that $\epsilon$ (which is a fixed constant) is present only in the polynomial on the input size term.

 The solution $X''$ is feasible to $\calI'_1$, hence using Lemma~\ref{lemma:fpt-as-trans-to-origin-simpler} we can construct a feasible solution to $\calI$ with utility equal to
 \begin{align*}
   u^*(\calI''_1)
   \geq \frac{u^*(\calI'_1)}{1+\epsilon} \geq \frac{u^*(\calI)}{1+\epsilon},
 \end{align*}
  where the first inequality follows from Lemma~\ref{lemma:fpt-as-apx-lose-simpler} and the second inequality follows from Lemma~\ref{lemma:fpt-as-utility-correspond-simpler}.
  This finishes the proof.
\end{proof}

\subsection{W-hardness for g with Utility Requirements}\label{sec:w-hardness-utility}

While we do not know whether \gpb is $\FPT$ wrt.\ $g$, here we partially resolve this question by providing a proof of $\W[1]$-hardness wrt. $g$ for a slightly more general problem that we call \ugpb.
\ugpb has the same input and output as \gpb with only one difference: additionally we define {\it the utility requirement function} $u \colon \calF\rightarrow \mathbb{N}$ which means how much utility we have to get from projects belonging to each project-group.
Let $a \colon P \rightarrow \{1,\dots,|\calV|\}$ be an approval score function, i.e., $a(p) = |\{ v \in \calV: p \in P_v \}|$.
Formally \ugpb is defined as follows.

\defproblem{\ugpb}
{A set $P$ of projects with their cost function $c \colon P \rightarrow \mathbb{N}$,
a set $\calV$ of voters with their approval ballots $\calE=\{P_v\subseteq P \colon v\in \calV\}$,
a family of groups of projects $\calF \subseteq 2^P$ with their budget function $b \colon \calF\rightarrow \mathbb{N}$
and utility requirement function $u \colon \calF\rightarrow \mathbb{N}$,
a global budget limit $B$,
and a desired utility value $u$.}
{Is there a set of projects $X \subseteq P$ such that $\sum_{v \in \calV}|P_v \cap X| \geq u$, $\sum_{p\in X}c(p)\leq B$, and for every set $F \in \calF$, $\sum_{p\in F\cap X}c(p)\leq b(F)$ and $\sum_{p\in F\cap X} a(p)\geq u(F)$?}

Note that we override the notation for $u$, i.e.,
$u(\cdot)$ is a function, $u$ is a desired utility value and hence we have $u(P) = u$.

We will construct an FPT-reduction from the \arcsupply problem that is $\W[1]$-hard wrt. the number of vertices in the input graph (even on planar graphs)~\cite{BodlaenderLP09}.
In \arcsupply we are given a simple directed graph (no self-loops, no parallel arcs).
For each vertex we are given a positive integer ({\it demand}) and for every directed edge $e$ we have a list $L_e$ of {\it supply pairs}, i.e.,
$L_e = \{ (x_1^e,y_1^e), \dots, (x_{|L_e|}^e,y_{|L_e|}^e) \}$, where $x_i^e, y_i^e \in \mathbb{N}_{>0}$.
The task is to choose exactly one supply pair from each of the lists $L_e$ in such a way that each vertex gets at least as much supply as its demand.
A supply pair $(x_i^e,y_i^e)$ for $e=(\alpha,\beta)$ gives supply $x_i^e$ to vertex $\alpha$ and $y_i^e$ to vertex $\beta$.
Formally the problem is defined as follows.

\defproblem{\arcsupply}
{A simple directed graph $G=(V,E)$,
a demand function $d \colon V \rightarrow \mathbb{N}_{>0}$,
a list of supply pairs $L_e = \{ (x_1^e,y_1^e), \dots, (x_{|L_e|}^e,y_{|L_e|}^e) \}$ for every $e \in E$,where $x_i^e, y_i^e \in \mathbb{N}_{>0}$ are encoded in unary.}
{Is there a function $S$ that for each $e \in E$ we have $S(e) \in \{ 1, \dots, |L_e| \}$ and
for every $\alpha \in V$ we have $d(\alpha) \leq \sum_{e=(\alpha,\beta) \in E} x_{S(e)}^e + \sum_{e=(\beta,\alpha) \in E} y_{S(e)}^e$.}

Without loss of generality we can assume that for every $e \in E$ we have $x_1^e > x_2^e > \cdots > x_{|L_e|}^e$ and $y_1^e < y_2^e < \cdots < y_{|L_e|}^e$.
We can satisfy this after a preprocessing step by removing {\em dominated} pairs and sorting the remaining ones, where
a pair $(x_i,y_i)$ dominates $(x_j,y_j)$ if $x_i \geq x_j$ and $y_i \geq y_j$.
The preprocessing step takes polynomial-time: $\Oh(|E|\cdot \max_{e \in E}|L_e|^2)$.

\begin{theorem}\label{thm:w-h-ugpb}
 \ugpb is $\W[1]$-hard wrt.\ $g$.
\end{theorem}
\begin{proof}
 We construct a reduction from \arcsupply which is $\W[1]$-hard parameterized by the number of vertices in the input graph~\cite{BodlaenderLP09}.
 Let $(G=(V,E), d, (L_e)_{e \in E})$ be an instance of \arcsupply with $x_1^e > x_2^e > \cdots > x_{|L_e|}^e$ and $y_1^e < y_2^e < \cdots < y_{|L_e|}^e$ for every edge $e \in E$.
 We define a large number $N = 2 \cdot \sum_{e \in E}\sum_{i=1}^{|L_e|} (x_i^e+y_i^e)$, and let $L$ be the input size of all lists $(L_e)_{e \in E}$.

 For every edge $e=(\alpha,\beta) \in E$ and every supply pair $(x_i^e,y_i^e) \in L_e$ we create two projects in $P$:
 project $p(e,\alpha,i)$ of cost $N-x_i^e$ and approval score $i$, and
 project $p(e,\beta,i)$  of cost $N-y_i^e$ and approval score $|L_e|-i+1$.

 In order to maintain the desired approval scores, we define voters $v_1, v_2, \dots, v_{\max_{e \in E}|L_e|}$, where a voter $v_i$ approves all projects with approval score equal-or-greater than $i$.
 Then indeed, a project $p$ is approved by exactly $v_1, v_2, \dots, v_{a(p)}$.

 Next we define groups of projects belonging to $\calF$ and their budget limits.
 For every vertex $\alpha \in V$ we create a group of projects $F_\alpha = \{ p(e,\alpha,i) \in P \}$ with its budget limit $b(F_\alpha) = N \cdot \deg(\alpha) - d(\alpha)$,
 where $\deg(\alpha)$ is a degree of $\alpha$ in $G$ ($\deg(\cdot)$ counts both in-edges and out-edges).
 Budget limit $b(F_\alpha)$ corresponds to supply demand for vertex $\alpha$.
 For every edge $e=(\alpha,\beta) \in E$ we create three groups of projects:
 $F_{e,\alpha} = \{ p(e,\alpha,i) \in P \}$,
 $F_{e,\beta } = \{ p(e,\beta,i ) \in P \}$ and
 $F_e = F_{e,\alpha} \cup F_{e,\beta}$.
 We define budget limits for them as follows:
 $b(F_{e,\alpha}) = b(F_{e,\beta}) = N$ and
 $b(F_e) = 2N$.
 Budget limit $b(F_{e,\alpha})$ corresponds to choosing at most one $x_i^e$ from $L_e$ and $b(F_{e,\beta})$ corresponds to choosing at most one $y_i^e$ from $L_e$.
 We define the global budget as $B=\sum_{p \in P} c(p)$.

 The last ingredient is to specify utility requirements for groups of projects.
 We only restrict utility requirement on $F_e$ for every $e \in E$ by setting $u(F_e) = |L_e|+1$.
 Utility requirement $u(F_e)$ together with dominating pairs property of $(x_i^e,y_i^e)$ give us possibility that even if in a solution we have $x_i^e$ and $y_j^e$ with $i \neq j$ we can replace them with $x_i^e$ and $y_i^e$ (which corresponds to a supply pair in the original instance).
 For the remaining groups utility requirement is $0$, also $u=u(P)=0$.

 The reduction takes time at most $\Oh(|E| \cdot L^2)$ for creating the projects and voters and $\Oh(|V||E|^2 \cdot L)$ to create groups together with budget limits and utility requirements.
 Hence, in total, the running time is bounded by $\Oh(|V||E|^2 \cdot L^2)$ that is polynomial in the input size.
 We have $|V|$ many groups of type $F_\alpha$, $2|E|$ many groups of type $F_{e,\alpha}$ and $|E|$ many groups of type $F_e$.
 Hence we get $g = |\calF| = |V|+3|E| \leq 4|V|^2 = \Oh(|V|^2)$.
 $g$ is bounded by a function of the parameter $|V|$ for which \arcsupply is $\W[1]$-hard~\cite{BodlaenderLP09}, hence this reduction is an FPT-reduction.
 Next we show correctness of the reduction.

 {\bf Correctness:}
 Let $S$ be a solution to an \arcsupply instance.
 We construct a solution to \ugpb by simply choosing the corresponding projects, i.e., for every edge $e = (\alpha,\beta) \in E$, we take to the solution $X$ projects $p(e,\alpha,S(e))$ and $p(e,\beta,S(e))$.

 Recall that desired total utility value is $0$ and global budget $B$ is large enough to allow any subset of projects be a solution.
 Therefore, we need to check only budget and utility requirements in the groups from $\calF$.

 For every $\alpha \in V$ we check budget limit for $F_\alpha$.
 Cost of the projects taken from $F_\alpha$ is equal to
 \begin{align*}
  \sum_{p \in F_\alpha \cap X} c(p)
  &= \sum_{e=(\alpha,\beta) \in E} \sum_{i: p(e,\alpha,i) \in X} (N-x_i^e)
   +\sum_{e=(\beta,\alpha) \in E} \sum_{i: p(e,\alpha,i) \in X}  (N-y_i^e)\\
  &= \sum_{e=(\alpha,\beta) \in E} (N-x_{S(e)}^e)
    +\sum_{e=(\beta,\alpha) \in E} (N-y_{S(e)}^e)
  = \deg(\alpha) \cdot N
    -\sum_{e=(\alpha,\beta) \in E} x_{S(e)}^e
    -\sum_{e=(\beta,\alpha) \in E} y_{S(e)}^e,
 \end{align*}
 and this is upper-bounded by $\deg(\alpha) \cdot N -d(\alpha) = b(F_\alpha)$ because $S$ is feasible.
 Therefore budget limit for $F_\alpha$ is kept.

 Utility requirement for $F_\alpha$ is $0$ hence always feasible.

 Next, for every edge $e=(\alpha,\beta) \in E$ we consider groups $F_{e,\alpha}, F_{e,\beta}$ and $F_e$.
 Cost of projects taken from $F_{e,\alpha}$ is equal to
 \begin{align*}
  \sum_{p \in F_{e,\alpha} \cap X} c(p) = \sum_{i: p(e,\alpha,i) \in X} (N-x_i^e) = N-x_{S(e)}^e
  \leq N = b(F_{e,\alpha}).
 \end{align*}
 Similarly we can show that the cost of projects taken from $F_{e,\beta}$ is at most $N = b(F_{e,\beta})$.
 It follows from the above and $F_e = F_{e,\alpha} \cup F_{e,\beta}$ that cost of projects taken from $F_e$ is at most $2N = b(F_e)$.

 For $F_{e,\alpha}$ and $F_{e,\beta}$ the utility requirement is $0$, hence always feasible.
 The utility requirement for $F_e$ is kept because
 \begin{align*}
  \sum_{p \in F_e \cap X} a(p) &= \sum_{p \in F_{e,\alpha} \cap X} a(p) + \sum_{p \in F_{e,\beta} \cap X} a(p)
  = \sum_{i: p(e,\alpha,i) \in X} i + \sum_{i: p(e,\beta,i) \in X} (|L_e|-i+1)
  = S(e) + |L_e|-S(e)+1 = u(F_e).
 \end{align*}

{\bf Soundness:}
Let $X$ be a solution to the \ugpb instance.
First we show that for every edge $e=(\alpha,\beta) \in E$
there are exactly one project from $F_{e,\alpha}$ and exactly one project from $F_{e,\beta}$ in $X$.

\begin{lemma}\label{lem:fea1}
 For every edge $e=(\alpha,\beta) \in E$ we have
 $$|X \cap F_{e,\alpha}|=|X \cap F_{e,\beta}|=1.$$
\end{lemma}
\begin{proof}
 Let $e=(\alpha,\beta)$ be any edge from $E$.
 We have $|X \cap F_{e,\alpha}| \leq 1$ because otherwise, if $|X \cap F_{e,\alpha}| > 1$ then
 $$ \sum_{p \in F_{e,\alpha} \cap X} c(p) \geq (N-x_1^e)+(N-x_2^e) > N = b(F_{e,\alpha})$$
 that is a contradiction with feasibility of $X$.
 Analogously we show that $|X \cap F_{e,\beta}| \leq 1$.

 If $|X \cap F_{e,\alpha}| = 0$ then utility we get in $F_e$ is equal to
 $$ \sum_{p \in F_e \cap X} a(p) = \sum_{p \in F_{e,\beta} \cap X} a(p)$$
 that is at most
 $$ |F_{e,\beta} \cap X| \cdot |L_e| \leq |L_e| < u(F_e)$$
 and this would be a contradiction with feasibility of $X$.
 Therefore $|X \cap F_{e,\alpha}| > 0$ and analogously we show that $|X \cap F_{e,\beta}| > 0$
 what finishes the proof.
\end{proof}

Due to Lemma~\ref{lem:fea1}, for every edge $e=(\alpha,\beta) \in E$
we define $i(e)$ being such that $\{ p(e,\alpha,i(e)) \} = X \cap F_{e,\alpha}$
and $j(e)$ being such that $\{ p(e,\beta,j(e)) \} = X \cap F_{e,\beta}$.

In every feasible solution to \ugpb there is a dependence between $i(e)$ and $j(e)$.

\begin{lemma}\label{lem:yij-monoton}
 For every edge $e\in E$ we have $y_{j(e)}^e \leq y_{i(e)}^e$.
\end{lemma}
\begin{proof}
 Let $e=(\alpha,\beta)$ be an edge from $E$.
 The utility obtained in group $F_e$ is equal to
\begin{align*}
   \sum_{p \in F_e \cap X} a(p) 
 = \sum_{p \in F_{e,\alpha} \cap X} a(p) + \sum_{p \in F_{e,\beta} \cap X} a(p)
 = a(p(e,\alpha,i(e))) + a(p(e,\beta,j(e)))
 = i(e) + (|L_e|-j(e)+1).
\end{align*}
As $X$ is a feasible solution to \ugpb we have $|L_e|+1 = u(F_e) \leq i(e) + (|L_e|-j(e)+1)$, hence $j(e) \leq i(e)$.
Together with $y_1^e < y_2^e < \cdots < y_{|L_e|}^e$ we get the required inequality.
\end{proof}

We construct a solution $S_X$ to \arcsupply setting $S_X(e) := i(e)$ for every edge $e \in E$.

We will show that $S_X$ is a feasible solution by arguing that every $\alpha \in V$ is supplied by at least its demand.
Let us fix $\alpha \in V$.
We know that the budget limit of $F_\alpha$ is kept, so we have
\begin{align*}
   N \cdot \deg(\alpha) - d(\alpha)
  = b(F_\alpha)
  \geq \sum_{p \in F_\alpha \cap X} c(p)
 &=    \sum_{e \in E} \sum_{p(e,\alpha,i) \in X} c(p(e,\alpha,i))\\
 &=    \sum_{e=(\alpha,\beta) \in E} \sum_{p(e,\alpha,i) \in F_{e,\alpha} \cap X} (N-x_i^e)
 +    \sum_{f=(\beta,\alpha) \in E} \sum_{p(f,\alpha,j) \in F_{f,\alpha} \cap X} (N-y_j^f)\\
 &=    \sum_{e=(\alpha,\beta) \in E} (N-x_{i(e)}^e)
  +    \sum_{f=(\beta,\alpha) \in E} (N-y_{j(f)}^f)\\
 &=    N \cdot \deg(\alpha)
  -    \sum_{e=(\alpha,\beta) \in E} x_{i(e)}^e
  -    \sum_{f=(\beta,\alpha) \in E} y_{j(f)}^f\\
 &\geq N \cdot \deg(\alpha)
  -    \sum_{e=(\alpha,\beta) \in E} x_{i(e)}^e
  -    \sum_{f=(\beta,\alpha) \in E} y_{i(f)}^f,
\end{align*}
where the last one inequality follows from Lemma~\ref{lem:yij-monoton}.
From this and definition of $S_X$ we get
$$ d(\alpha)
   \leq  \sum_{e=(\alpha,\beta) \in E} x_{S_X(e)}^e
        +\sum_{f=(\beta,\alpha) \in E} y_{S_X(f)}^f.$$
so the demand requirement for $\alpha$ is kept in a solution $S_X$.
\end{proof}

\section{Approximation Algorithms}\label{sec:apx}

Recall that an algorithm for \mgpb has an approximation factor $\alpha$, $\alpha \geq 1$,
if it always outputs a solution with at least an $\alpha$ fraction of the utility achieved by an optimal solution.

In the Related Work subsection we showed that an instance of \mgpb is an instance of \xdk{(g+1)}
where each project $p \in P$ has only two possible costs $\{0,c(p)\}$ over $g+1$ dimensions.
There is a polynomial-time approximation scheme (PTAS) for \ddk when $d$ is a constant,
running in $\Oh(m^{\lceil d/\epsilon\rceil-d})$-time~\cite[Theorem 9.4.3]{knapsackbook};
hence we have a PTAS for \gpb with $g=\Oh(1)$, running in $\Ohstar(m^{\Oh(1/\epsilon)})$-time.
We can improve this result using our $\XP$ wrt.\ $g$ algorithm from Theorem~\ref{thm:xp-g},
as for $g=\Oh(1)$, the algorithm provides an exact solution in polynomial-time.
\begin{corollary}\label{cor:g-constant-poly}
 For $g=\Oh(1)$, \gpb can be solved in polynomial-time.
\end{corollary}

In contrast, there is no FPTAS for \xdk{2} unless $\P=\NP$~\cite[Theorem 9.4.3]{knapsackbook}.
In that proof the authors use an $\NP$-hardness reduction from the {\sc Partition} problem to \xdk{2} with two different costs of a project $p$ in the dimensions, i.e.,
$c(p)$ and $(\max_{r \in P} c(r))-c(p)$.
Furthermore, unless $\FPT=\W[1]$, there is no efficient polynomial-time approximation scheme (EPTAS) for \xdk{2}~\cite{KulikS10}.

Next we focus on non-constant values of $g$.
\ddk can be approximated in polynomial time up to $d+1$ factor~\cite[Lemma 9.4.2]{knapsackbook}, thus we have the following.
\begin{corollary}\label{cor:g-plus-2-apx}
 There exists a polynomial-time $(g+2)$-approximation algorithm for \mgpb.
\end{corollary}

While seemingly very weak, this approximation guarantee is almost tight.

\begin{theorem}\label{thm:sqrt-g-apx}
 For any $\epsilon>0$, there does not exist a polynomial-time $(g^{1/2-\epsilon})$-approximation algorithm for \mgpb unless $\P=\NP$ even if $g \leq m^2$.
\end{theorem}
\begin{proof}

 We rely on the inapproximability of \maxis (\maxissmall):
 for any $\epsilon'>0$, \maxissmall is $\NP$-hard to approximate to within $|V|^{1-\epsilon'}$ factor~\cite{Zuckerman07}.
 Fix $\epsilon>0$, and use the reduction from Theorem~\ref{thm:w-hard-app2} with setting the global budget to be $B=m$.
 This reduction is preserving the approximation factor, i.e.,
 an $\alpha$-approximate solution to the resulting \mgpb instance is also an $\alpha$-approximate solution to \maxis.
 This is so as any feasible solution to \mgpb of total utility $u$ contains exactly $u$ projects that correspond to an $u$-element subset of $V$,
 which is an independent set
 (from every edge we choose at most one vertex, as otherwise, when choosing both endpoints, we would violate a group budget constraint).
 Assume there exists a $(g^{1/2-\epsilon})$-approximation algorithm for \mgpb.
 Then, we also have $(g^{1/2-\epsilon})$-approximation algorithm for \maxissmall.
 We have $g=|E| \leq |V|^2$, hence
 $g^{1/2-\epsilon} = g^{(1-2\epsilon)/2} \leq |V|^{1-2\epsilon}$, which means that
 we achieved $|V|^{1-\epsilon'}$-approximation algorithm for \maxissmall for $\epsilon' = 2\epsilon>0$,
 in contradiction to the hardness of approximation of \maxissmall~\cite{Zuckerman07}.
\end{proof}

As we have $m = |V|$ in the proof above, we can also derive $(m^{1-\epsilon})$-hardness of approximation (unless $\P=\NP$).
One may consider that \mgpb is easier if $g = o(m^2)$, but even if $g$ is linear on $m$ we can exclude a PTAS:
\begin{theorem}\label{thm:no-ptas-g-linear-on-m}
 Assuming $\P \neq \NP$, there does not exist a PTAS for \mgpb, even if $g \leq \frac{3}{2}m$.
\end{theorem}

\begin{proof}
 We will use $\APX$-completeness of \maxissmall on 3-regular 3-edge-colorable graphs~\cite{chlebik2003approximation}
 and the reduction from Theorem~\ref{thm:paraNP-l+s+f}
 in which we are changing a global budget into $B=m$.
 Analogously to the proof of Theorem~\ref{thm:sqrt-g-apx}
 this reduction is approximation preserving.
 It means that \mgpb is $\APX$-hard, which excludes a PTAS assuming $\P \neq \NP$.
 Additionally, in the reduction we have $g=|E|=\frac{3}{2}|V|=\frac{3}{2}m$ because the input graph is 3-regular.
\end{proof}

What happens for $g = o(m)$? Here is a partial answer that uses our FPT approximation scheme wrt.\ $g$ (Theorem~\ref{thm:fpt-as-simpler}).

\begin{theorem}\label{thm:ptas-for-g-loglog-i}
 Let $|\calI|$ be the size of an instance $\calI$ of \mgpb.
    There exists a PTAS for \mgpb if $g \leq \log_4\log(|\calI|^{\Oh(1)})$.
\end{theorem}
\begin{proof}
 For a fixed $\epsilon>0$ the running time of the FPT approximation scheme is bounded by
 $\Ohstar(2^{\frac{1}{2} \cdot 4^g})$ (see Lemma~\ref{lem:fpt-as-time-simpler}).
 Hence, putting there $g \leq \log_4\log(|\calI|^{\Oh(1)})$ we get a following bound on the running time:
 \begin{align*}
  \Ohstar(2^{\frac{1}{2} \cdot 4^{\log_4\log(|\calI|^{\Oh(1)})}}) = \Ohstar(2^{\frac{1}{2} \cdot \log(|\calI|^{\Oh(1)})}) =
  \Ohstar(|\calI|^{\Oh(1)}) = \Ohstar(1),
 \end{align*}
 that is polynomial for a fixed $\epsilon$ hence we get a PTAS.
\end{proof}

As $m$ is bounded by the input size we get the following.

\begin{corollary}\label{cor:ptas-for-g-loglogm}
 There exists a PTAS for \mgpb if $g \leq \log_4\log(m^{\Oh(1)})$.
\end{corollary}

Our approximability and inapproximability results are summarized in Tables~\ref{tab:apx-results} and \ref{tab:inapx-results}, respectively.

\section{Outlook}\label{sec:outlook}

Motivated by PB scenarios in which it is useful to consider geographic constraints and thematic constraints, we enriched the standard approval-based model of PB
(in particular, the model of \emph{Combinatorial PB}~\cite{aziz2021participatory}),
by introducing a group structure over the projects and requiring group-specific budget limits for each group.

We have showed that, while being computationally intractable in general, the enriching PB instances with such group structure and its corresponding budget constraints comes at essentially no computational cost
if there are not so many such groups or if the structure of these groups is hierarchical or close to being such.
We complemented our analysis with lower bounds and approximation algorithms.

Practically, while our focus is on a theoretical understanding of the combinatorial structure of our problems,
some of our results are already showing efficient complexity and thus can be used practically as they are.
Other results giving an evidence of being in $\FPT$ or $\XP$, while having rather high complexity (e.g., double exponential dependency on the relevant parameter),
show nevertheless that efficient algorithms may exist,
thus more research might be instructive in finding algorithms with even better running time.

Some future research directions are the following:
\begin{enumerate}

\item
\textbf{Complexity wrt.\ $g$:}
We have one main open question remaining, namely, the parameterized complexity of \gpb wrt.\ $g$, that is worth settling.

\item
\textbf{Further parameters:}
Considering more parameters, including parameter combinations, would be natural for future research.
Furthermore, our reduction from \is in Theorem~\ref{thm:paraNP-l+s+f} can give also time lower-bounds under the Exponential Time Hypothesis
(see, e.g., a survey by Lokshtanov \textit{et al.}~\cite{LokshtanovMS11}).
Generally speaking, it would be interesting to evaluate what are limitations of parameterized algorithms for \gpb
(in particular with respect to parameter $g$)
and how our current results are close to such lower bounds.

\item
\textbf{Different aggregation methods:}
Another direction would be to consider other functions of aggregating utilities of voters.
For example, instead of maximizing the \emph{sum} over voter utilities, one might be egalitarian and aim at maximizing the \emph{minimum} over voter utilities.

\item
\textbf{Cardinality constraints:}
Another direction would be to consider additional cardinality constraints.
While all hardness results hold for the more general variant where we consider both knapsack \emph{and} cardinality constraints,
a careful analysis shall be taken to explore which of the tractability results hold for this setting as well.

\item
\textbf{Project interactions:}
Another direction would be to consider interactions among projects, such as substitution and complementarities.
These can be modeled via general functions, corresponding to the utility of voters from a certain number of projects funded and approved from each group,
in the spirit of the work of Jain \textit{et al.}~\cite{pbsub}.
\end{enumerate}

\section*{In Memory of Rolf Niedermeier}

Nimrod would like to speak for the co-authors regarding Rolf:
First of all, as my PhD supervisor, Rolf had tremendous influence on my scientific career.
Perhaps one of the most magical features of Rolf was his ability to connect people.
In a way, this paper is a result of this ability.
In particular, for my PhD defense, Rolf invited Hadas Shachnai (from the Technion) to serve in my PhD committee. Rolf mentioned that Hadas was at the time the PhD advisor of Meirav Zehavi, and he told me that I should perhaps contact Meirav as she is ``doing good research''. 
Well, it took me some years, but eventually I contacted Meirav Zehavi and this paper is the result of this connection.

\section*{Acknowledgements}

We would like to thank the anonymous reviewers for their helpful comments.

Pallavi Jain was supported by IIT Jodhpur seed grant (Grant No. I/SEED/PJ/20210119) and SUPRA grant by SERB, India (Grant No. SPR/2021/000860).
Krzysztof Sornat was partially supported by the European Research Council (ERC) under the European Union’s Horizon 2020 research and innovation programme (grant agreement No. 101002854) and
by the SNSF Grant 200021\_200731/1.
Nimrod Talmon was supported by the Israel Science Foundation (ISF; Grant No. 630/19).
Meirav Zehavi was supported by the European Research Council (ERC) grant titled PARAPATH (grant No. 101039913).

\bibliographystyle{alpha}
\bibliography{bib}

\end{document}